\documentclass[11pt]{article}
\pdfoutput=1
\usepackage[utf8]{inputenc}
\usepackage{microtype}
\usepackage[sort,numbers]{natbib}
\usepackage[bookmarks=false,pdfpagelabels=false]{hyperref}
\usepackage[textsize=scriptsize]{todonotes}
\usepackage{amsmath}
\usepackage{amssymb}
\usepackage{amsthm}
\usepackage{mathtools}
\usepackage{xcolor}
\usepackage{paralist}
\usepackage{tikz}
\usepackage{aliascnt}
\usetikzlibrary{decorations.text,fadings}
\usepackage{verbatim}
\usetikzlibrary{shapes,positioning,decorations.pathreplacing}
\usepackage{subfigure}
\hypersetup{breaklinks,pdfdisplaydoctitle,colorlinks,linkcolor=red!50!black,citecolor=green!50!black}
\usepackage{authblk}

\author[1]{René van Bevern}
\author[2]{Rodney G. Downey}
\author[3]{Michael R. Fellows}
\author[4]{Serge~Gaspers}
\author[3]{Frances~A.~Rosamond}

\date{}

\affil[1]{Institut f\"ur Softwaretechnik und Theoretische Informatik,
  TU Berlin, Germany, \texttt{rene.vanbevern@tu-berlin.de}}

\affil[2]{Victoria University of Wellington, New Zealand,
  \texttt{rod.downey@vuw.ac.nz}}

\affil[3]{School of Engineering and IT, Charles Darwin University,
  Darwin, Australia,
  \texttt{\{michael.fellows,frances.rosamond\}@cdu.edu.au}}

\affil[4]{The University of New~South~Wales and NICTA, Sydney,
  Australia, \texttt{sergeg@cse.unsw.edu.au}}

\newcommand{\mcF}{\mathcal F}
\newcommand{\CNFSAT}{\textsc{CNF-Sat}}

\newcommand{\joint}{joint}
\newcommand{\coloneqq}{:=}
\newcommand{\mytheta}[3]{\mathop{\mathrm{Jcut}}\nolimits^{#1}_{#2}(#3)}
\newcommand{\normtheta}[3]{\mathop{\mathrm{Cut}}\nolimits^{#1}_{#2}(#3)}

\newcommand{\full}{ge\-ne\-ra\-tor-to\-tal}
\newcounter{mycounter}

\title{Myhill-Nerode methods for hypergraphs\thanks{A preliminary
    version of this article appeared in the proceedings of ISAAC
    2013~\cite{BevernFGR13}. This extended and revised version
    contains the full proof details, more figures, and corollaries to
    make the application of the Myhill-Nerode theorem for hypergraphs
    easier in an algorithmic setting.  Moreover, it provides a fix to
    the proof of the Myhill-Nerode theorem for graphs in the books of
    \citet{DF99,DF13}}}

\newcommand{\hgcw}{\textsc{Hy\-per\-graph Cut\-width}}
\newcommand{\ghtw}{\textsc{Generalized Hypertree Width}}
\newcommand{\htw}{\textsc{Hypertree Width}}
\newcommand{\fhtw}{\textsc{Fractional Hypertree Width}}
\newcommand{\cw}{\textsc{Cutwidth}}

\newcommand{\brc}[1]{\textsc{$#1$-HCW}}
\newcommand{\kghtw}[1]{\textsc{$#1$-GHTW}}
\newcommand{\khtw}[1]{\textsc{$#1$-HTW}}
\newcommand{\kfhtw}[1]{\textsc{$#1$-FHTW}}

\newcommand{\iw}{incidence treewidth}

\newcommand{\br}{incidence graph}

\newcommand{\crc}{canonical right congruence}

\newcommand{\bndg}[1]{$#1$-bound\-a\-ried graph}
\newcommand{\bndh}[1]{$#1$-bound\-a\-ried hypergraph}
\newcommand{\simtest}{\ensuremath{\sim_{\mathcal T}}}

\newcommand{\cmax}{\ensuremath{c_{\mathrm{max}}}}
\newcommand{\ularge}[1]{\ensuremath{\mathcal U^{\smash{\text{large}}}_{t,#1}}}
\newcommand{\ulargep}{\ensuremath{\mathcal U^{\smash{\text{large}}}_{t}}}
\newcommand{\ularget}[1]{\ensuremath{\mathcal U^{\smash{\text{large}}}_{2,#1}}}
\newcommand{\usmall}[1]{\ensuremath{\mathcal U^{\smash{\text{small}}}_{t,#1}}}
\newcommand{\usmallt}[1]{\ensuremath{\mathcal U^{\smash{\text{small}}}_{2,#1}}}
\newcommand{\usmallp}{\ensuremath{\mathcal U^{\smash{\text{small}}}_{t}}}
\newcommand{\hsmall}{\ensuremath{\mathcal H^{\smash{\text{small}}}_{t}}}
\newcommand{\hlarge}{\ensuremath{\mathcal H^{\smash{\text{large}}}_{t}}}
\newcommand{\hyperg}[1]{\ensuremath{\mathcal H(#1)}}
\newcommand{\wcg}[1]{$t$-boundaried hypergraph generator}
\newcommand{\passes}[1]{$k$-passes}

\DeclareMathOperator\Pos{Pos}
\DeclareMathOperator\lab{col}
\DeclareMathOperator\inc{inc}
\DeclareMathOperator\adj{adj}
\DeclareMathOperator\oplusc{\oplus_c}
\DeclareMathOperator\oplush{\oplus_h}

\DeclareMathOperator\congr{\sim_{\brc{k}}}

\newtheorem{theorem}{Theorem}[section]

\newaliascnt{conjcnt}{theorem}
\newaliascnt{excnt}{theorem}
\newaliascnt{obscnt}{theorem}
\newaliascnt{constrcnt}{theorem}
\newaliascnt{clacnt}{theorem}
\newaliascnt{lemcnt}{theorem}
\newaliascnt{defcnt}{theorem}
\newaliascnt{colcnt}{theorem}

\newtheorem*{dftheorem}{Myhill-Nerode Theorem for Graphs}
\newtheorem*{mntheorem}{Myhill-Nerode Theorem}
\newtheorem{corollary}[colcnt]{Corollary}

\newtheorem{conjecture}[conjcnt]{Conjecture}

\theoremstyle{definition}
\newtheorem{observation}[obscnt]{Observation}
\newtheorem{definition}[defcnt]{Definition}
\newtheorem{construction}[constrcnt]{Construction}
\newtheorem{cla}[clacnt]{Claim}
\newtheorem{lemma}[lemcnt]{Lemma}

\theoremstyle{remark}
\newtheorem{example}[excnt]{Example}

 \newcommand{\decprob}[4]{%
 {\def\descriptionlabel##1{\hspace\labelsep\quad{}\it{}##1}%
 \par\vspace{\topsep}\noindent%
 \begin{compactdesc}
 \item[\textsc{#1}]
 \item[Input:] #2
 \item[Question:] #3
\end{compactdesc}
}\vspace{\topsep}}

\tikzstyle{red}=[circle,draw=black,fill=gray!20,minimum size=5mm, inner sep=1pt]
\tikzstyle{blue}=[circle,draw=black,fill=white,minimum size=5mm, inner sep=1pt]
\pgfdeclarelayer{background}
\pgfsetlayers{background,main}

\makeatletter
\def\NAT@spacechar{~}%
\makeatother

\begin{document}
\maketitle
\vspace{-2em}
\paragraph{Abstract.} We give an analog of the Myhill-Nerode methods
from formal language theory for hypergraphs and use it to derive the
following results for two NP-hard hypergraph problems.
  \begin{itemize}
  \item We provide an algorithm for testing whether a hypergraph has
    cutwidth at most~$k$ that runs in linear time for constant~$k$. In
    terms of parameterized complexity theory, the problem is
    \emph{fixed-parameter linear} parameterized by~$k$.

  \item We show that it is not expressible in monadic second-order
    logic whether a hypergraph has bounded (fractional, generalized)
    hypertree width.  The proof leads us to conjecture that, in terms
    of parameterized complexity theory, these problems are W[1]-hard
    parameterized by the \emph{incidence treewidth} (the treewidth of
    the incidence graph).
  \end{itemize}
  \noindent Thus, in the form of the Myhill-Nerode theorem for
  hypergraphs, we obtain a method to derive linear-time algorithms and
  to obtain indicators for intractability for hypergraph problems
  parameterized by incidence treewidth.

  In an appendix, we point out an error and a fix to the proof of the
  Myhill-Nerode theorem for graphs in Downey and Fellow's book on
  parameterized complexity.

% \keywords{NP-hard problems\and fixed-parameter algorithms\and automata
%   theory\and cutwidth \and hypertree width}

\section{Introduction}\label{sec:intro}
There are two prevalent algorithmic techniques for solving NP-hard
problems in linear time on graphs of bounded treewidth---a measure
for the ``tree-likeness'' of a graph:
{\begin{description}
\item[Technique~1.] Compute a \emph{tree decomposition}---a tree-like
  representation---of the input graph in linear time~\citep{Bod96} and
  use dynamic programming from the leaves to the root of the tree
  decomposition.
\item[Technique~2.] Express the graph property to be decided in
  monadic second-order logic of graphs; the expression can be turned
  into a linear-time algorithm deciding the graph
  property~\citep[Theorem~6.4(1)]{CE12}.
\end{description}
For a primer on these algorithmic techniques, we refer to
\citet[Chapter~10]{Nie06}.

In some cases, graph problems do not easily give in to these standard
techniques. A third technique helps finding linear-time algorithms on
graphs of bounded treewidth or to prove the inapplicability of the
above standard techniques~\citep{AbrahamsonF91, BFW92, GNNW13}:
similarly to how regular languages can be recognized by finite
automata, some graph problems on graphs of bounded treewidth can be
solved in linear time by tree automata~\citep[Section~12.7]{DF13}. In
fact, many of the dynamic programming algorithms on tree
decompositions used in Technique~1 are based on a standard approach
that mimics tree automata~\citep{BFW92}. Moreover, Technique~2 is
based on the fact that an expression in monadic second-order logic can
be turned into a tree automaton~\citep[Chapter~6]{CE12}. Disproving
the existence of a tree automaton for a problem therefore shows that
it is presumably not straightforward to solve the problem on graphs of
bounded treewidth using Technique~1 and even impossible using
Technique~2.

A sufficient and necessary condition for the existence of a tree
automaton deciding some graph problem can be given by an adaption of the
Myhill-Nerode theorem from formal language theory to
graphs~\citep[Section~12.7]{DF13}, which helped gain insight into
the following graph problems:

\begin{description}
\item[\textsc{Cutwidth}.] Testing a graph for bounded cutwidth can be done
  in linear time~\citep{AbrahamsonF91}. \citet*{ThilikosSB05} later
  gave a dynamic programming algorithm that is significantly more
  technical, but has the advantage of constructing a solution instead
  of only answering whether a solution exists.

\item[\textsc{Bandwidth}.] The graph property of having bounded
  bandwidth is not recognizable by a tree
  automaton~\citep{AbrahamsonF91}.  Note that this unconditional result
  significantly strengthens the previously known NP-hardness of the
  problem on trees \citep{GGJK78} (which of course also excludes
  finite-state solvability of the problem for trees, but only under the
  assumption P${}\ne{}$NP).

\item[\textsc{Triangulating Colored Graphs}.] A tree automaton cannot
  decide whether a colored graph can be triangulated in such a way
  that adjacent vertices have distinct colors~\citep{BFW92}. The
  problem is known as \textsc{Perfect Phylogeny} in the context of
  molecular biology and later turned out to be
  W[1]-hard~\citep{BFH94,BFHW00} parameterized by the treewidth~$t$, that
  is, not solvable in $O(n^c)$~time for any constant~$c$ independent
  of~$t$ under the widely accepted parameterized complexity assumption
  FPT${}\neq{}$W[1].
\end{description}

\noindent Our work extends the graph-theoretic analog of the
Myhill-Nerode characterization of regular languages to hypergraphs. In
this way, we provide a method to derive linear-time algorithms (or to
obtain an indication for intractability) for hypergraph problems on
hypergraphs with bounded \emph{\iw} (treewidth of the incidence
graph).  Thus, our work is tightly connected to the existence of
\emph{fixed-parameter algorithms}---a~rising technique that allows for
solving NP-hard problems exactly \emph{and} efficiently when certain
parameters of the input data are small~\cite{DF13, Nie06, FG06}. From
this point of view, \iw{} is an interesting hypergraph parameter,
since the \iw{} of a hypergraph is not greater than the treewidth of
its primal or dual graph (two commonly used treewidth generalizations
for hypergraphs) but can be arbitrarily
smaller~\cite{KolaitisV00,SamerS10}.

Applying Myhill-Nerode methods to hypergraphs, we obtain results for
the problems \hgcw{} and \textsc{(Generalized, Fractional) Hypertree
  Width}, which will be formally defined in \autoref{sec5} and
\autoref{sec6}, respectively. 

\subsection{Related work}

\paragraph{Generalizations of the Myhill-Nerode theorem.}
The Myhill-Nerode theorem as sufficient and necessary condition for a
formal language being regular is due to \citet{Myh57} and
\citet{Ner58}. Since then, analogs of the Myhill-Nerode theorem were
provided for graphs of bounded treewidth~\citep{AbrahamsonF91},
matroids of bounded branchwidth~\citep{Hli05}, graphs of bounded
rankwidth~\citep{GH10}, and edge- and vertex-colored graphs of bounded
treewidth and cliquewidth~\citep[Sections~4.2.2 and~4.4.2]{CE12}.

\paragraph{Finite-state approaches to solving graph problems for graphs
  of bounded pathwidth and treewidth.}

The earliest, and seminal work, applying ideas from finite automata
theory to graph problems was the \citeyear{BernLW85} paper of
\citet*{BernLW85} based on $k$-terminal recursively defined families
of graphs.  These ideas were quickly taken up and extended in many
directions by many different groups of researchers with overlapping as
well as independent results obtained in a flurry of activity.  This
includes the influential \citeyear{WimerHL85} work of
\citet*{WimerHL85} (see also \citeauthor{Wimer87}'s \citeyear{Wimer87}
Ph.D. Thesis~\citep{Wimer87}), and \citet{MahajanP94}.  Finite-state
and Myhill-Nerode-related methods were explored in regards of
computing minor order obstruction sets by \citet{FLfocs89} and by
\citet{LagergrenA91}.  The regularity (in the sense of
\citet*{BernLW85}---finite-state dynamic programming
multiplication tables) of bounded treewidth and pathwidth was shown by
\citet{BodlaenderK96} and, independently, by \citet{LagergrenA91}.

\looseness=-1 In many cases, this early work circulated in some form (e.g.,
Technical Reports) years in advance of its eventual publication,
making the historical record murky---but it was an exciting time for
bounded treewidth and pathwidth algorithmics.  The period is ably
surveyed by \citet*{BoriePT08}, where many more references to early
work in the area can be found.

\paragraph{Communication complexity and generalizations of the
  Myhill-Ne\-rode theorem.}  Communication complexity was introduced by
\citet{Yao79} and measures the amount of information needed to be
transferred between two processors for computing a function~$f(x,y)$
when one processor receives~$x$ and the other processor receives~$y$.  A
pioneering and somewhat overlooked \citeyear{LW86} paper by \citet{LW86}
investigated the communication complexity of graph problems by studying
the following question: assume that $G$~is a graph that can be obtained
by ``gluing'' together two graphs~$G_1$ and~$G_2$ along a ``boundary''
of $t$~vertices,  what is the minimum amount~$f(t)$ of information
needed to be exchanged between two processors for deciding whether~$G$
has a certain property (for example, being Hamiltonian) when one
processor receives~$G_1$ and the other processor receives~$G_2$ as
input?

\looseness=-1 Although the work of \citet{LW86} is completely
unrelated to graphs of bounded treewidth, their notion is exactly what
in our article is called the \emph{large universe} of $t$-boundaried
graphs that can be glued together using an operator~$\oplusc$ (for the
precise definitions we refer to \autoref{colorglue} and
\autoref{def:ulargesmall} in \autoref{sec:mnhg}).  Therefore, %
the Myhill-Nerode approach yields insights into the
communication complexity of graph problems not only in the world of
graphs of bounded treewidth.  More specifically, it allows
for proving or disproving that the minimum amount of information
required to be transferred can be bounded by a function that only
depends on the boundary size~$t$.
However, the Myhill-Nerode approach is not limited to the
investigation of graph problems: in full generality,
assume that one has
\begin{enumerate}[1)]
\item a universe~$\mathcal{U}$ of mathematical objects of whatever sort,
  on which there is a partially defined operation $\mu\colon \mathcal{U}
  \times \mathcal{U} \rightarrow \mathcal{U}$ (sometimes called a
  \emph{partial groupoid}), and
\item a property~$\mathcal{P} \subseteq \mathcal{U}$ of interest of
  these objects.
\end{enumerate}
Then, one can define the \emph{canonical Myhill-Nerode equivalence
  relation $\sim_{\mathcal{P}}$} induced by~$\mathcal{P}$ on
$\mathcal{U}$ mimicking the formal language setting (there, $\mathcal{U}
= \Sigma^{*}$ and $\mu$ is string concatenation): $x \sim_{\mathcal{P}}
y$ if and only if, for all $z \in \mathcal{U}$, $\mu(x,z) \in
\mathcal{P}$ if and only if $\mu(y,z) \in \mathcal{P}$ (assuming $\mu$
is defined in both cases).  The analogy to the formal language setting
naturally leads to following interesting question: for which properties
(or \emph{classes of properties}) $\mathcal{P}$ does the canonical
equivalence relation~$\sim_{\mathcal P}$ have a finite number of
equivalence classes?

This abstract perspective often turns out to have powerful and elegant
algorithmic connections, \emph{as well as being of intrinsic interest in
  itself}.  For example, it is intrinsically interesting that if
$\mathcal{U}$ is the (large) universe of arbitrary $t$-boundaried graphs
(of unbounded treewidth) and $\mu$ is the $\oplusc$ gluing operation
defined in \autoref{colorglue}, then for any fixed~$t$ and any graph
property~$\mathcal{P}$ describable in monadic second-order logic, the
canonical equivalence relation~$\sim_{\mathcal P}$ has a finite number
of equivalence classes.  This statement is stronger than that of
Courcelle's theorem~\citep[Theorem~6.3(2)]{CE12}, which proves the
statement for the universe of graphs of treewidth at most~$t$, and was
first proved in the \citeyear{FellowsA89} manuscript of
\citet{FellowsA89}.  

The proof is exposed in later monographs of~\citet{DF99,DF13} and
exploits induction on the formula structure as well as the \emph{method
  of test sets}, which we will in the following apply also to hypergraph
problems.  Note that it is not obvious how to prove the above statement
in full generality using other techniques that are frequently applied to
solve graph problems on graphs of bounded treewidth, like dynamic
programming or careful bookkeeping about partial solutions, since
\emph{there is nothing to dynamically program on} in the universe of
\emph{unbounded} treewidth graphs.  Hence, the method of test sets seems
to be essential if one is interested in results related to communication
complexity (as we are, secondarily).  Some further discussion of these
different approaches and their virtues and weaknesses can be found in
the concluding section of this article.

\paragraph{Hypergraph Cutwidth.} \hgcw{} is a natural generalization of
the NP-complete~\cite{Gavril77} \cw{} problem and asks whether a
hypergraph has cutwidth at most~$k$.  For a formal definition, we refer
to \autoref{sec5}. In the context of VLSI design, \hgcw{} is known as
\textsc{Board Permutation}~\citep{MillerS91}.  Moreover, \hgcw{}
naturally arises in solving \CNFSAT{} in the context of automatically
testing digital hardware~\citep{PCK99,WCYK01}.

For the special case of \cw{} on graphs, several fixed-parameter
algorithms are known~\cite{AbrahamsonF91, FellowsL92, FellowsL94,
  ThilikosSB05, BodlaenderFT09}.  \citet{CahoonS83} showed algorithms
for \hgcw{} with~${k\le 2}$ running in $O(n)$~time for~$k=1$ and running
in $O(n^3)$~time for~$k=2$ on $n$-vertex hypergraphs.  For
arbitrary~$k$, \citet{MillerS91} designed an algorithm running in
$O(n^{k^2+3k+3})$~time.  Moreover, \citet{Nagamochi12} presented a
framework for solving cutwidth-related graph problems in
$n^{O(k)}$~time.

\paragraph{Hypertree Width.} \looseness=-1 \htw{}, \ghtw{}, and \fhtw{}
are the problems of checking whether a hypergraph has (generalized,
fractional) hypertree width~$k$. All three measures are generalizations
of treewidth to hypergraphs and formally defined in
\autoref{sec6}. %
It is known that \htw{} is W[2]-hard parameterized
by~$k$~\citep{GottlobGMSS05} and that \ghtw{} remains NP-hard even
for~$k=3$~\citep{GottlobMS09}. \citet{Marx10} expects \fhtw{} also to
be NP-hard for constant~$k$.  Hence, the computation of these width
parameters is presumably not fixed-parameter tractable parameterized
by~$k$ (that is, presumably not solvable in $n^c$~time for any
constant~$c$ independent of~$k$). Hence, it makes sense to investigate
whether the problems are fixed-parameter tractable with respect to
larger parameters~\citep{Nie10, KN12, FJR13}, like \iw{}.

\subsection{Our results and organization of this paper}
In \autoref{sec:prelims}, we
introduce the necessary graph and hypergraph notation, formally define
treewidth, \iw{}, and tree automata.

\looseness=-1 In \autoref{sec:mnhg}, we prove a Myhill-Nerode theorem
for hypergraphs of bounded \iw{}. Moreover, the section discusses how
the Myhill-Nerode theorem for hypergraphs yields linear-time
algorithms and excludes the possibility for monadic second-order logic
expressions for hypergraph problems.

In \autoref{sec5}, we exploit the Myhill-Nerode theorem for
hypergraphs to show that \hgcw{} can be solved in $O(n+m)$~time for
constant~$k$, thus showing \hgcw{} to be \emph{fixed-parameter linear}
parameterized by~$k$.

\looseness=-1 In \autoref{sec6}, we exploit the Myhill-Nerode theorem
to show that \htw{}, \ghtw{}, and \fhtw{} are not decidable by a
finite tree automaton and, hence, not expressible in monadic
second-order logic. Moreover, we obtain an indication that they are
not fixed-parameter tractable parameterized by \iw{}.

\section{Preliminaries}\label{sec:prelims}

\paragraph{Graphs and hypergraphs.}
A \emph{hypergraph}~$H$ is a pair~$(V,E)$, where $V(H):=V$~is a set of
\emph{vertices} and $E(H):=E$~is a set of \emph{hyperedges} such that
$e\subseteq V$ for each~$e\in E$. In this work, we allow $E$~to be a
multiset and there may be singleton and empty hyperedges. If not
stated otherwise, we use $n\coloneqq{}|V|$ and~$m\coloneqq{}|E|$. Two
hypergraphs $G$~and~$H$ are isomorphic and we write $G\cong H$ if
there is a bijection~$f\colon V(G)\to V(H)$ such that $e$~is an edge
with multiplicity~$i$ of~$G$ if and only if~$\{f(v)\mid v\in e\}$ is
an edge with multiplicity~$i$ of~$H$. The bijection~$f$ is called
\emph{(hypergraph) isomorphism}.

A \emph{graph} is a hypergraph in which every edge has cardinality two
and is present at most once.  Two vertices~$v,w\in V$ are
\emph{adjacent} or \emph{neighbors} if~$\{v,w\}\in E$.  The
\emph{(open) neighborhood}~$N_G(v)$ of a vertex~$v\in V$ in a
graph~$G$ is the set of vertices that are adjacent to~$v$%
. If the graph~$G$ is clear from the context, we drop the subscript~$G$.
A subset~$S\subseteq V$ is an \emph{independent set} if
no two vertices in~$S$ are adjacent in~$G$.

The \emph{primal graph} of a hypergraph~$H$, denoted $\mathcal{G}(H)$,
is the graph with vertex set~$V$ that has an edge $\{u, v\}$ if
$u$~and~$v$ are together in some hyperedge in~$H$.  It is sometimes
called the Gaifman graph of~$H$. The \emph{incidence graph} of a
hypergraph~$H$, denoted $\mathcal{I}(H)$, is the bipartite
graph~$(V',E')$ with vertex set $V'=V \cup E$ and such that, for
each~$v\in V$ and~$e\in E$, there is an edge~$\{v, e\} \in E'$ if and
only if~$v\in e$.

\paragraph{Graph decompositions.}
A \emph{tree decomposition}~$(T,\beta)$ for a graph~$G=(V,E)$ consists
of a rooted tree~$T$ and a mapping~$\beta\colon T\to 2^V$ of each
\emph{node}~$x$ of the tree~$T$ to a subset~$V_x:=\beta(x)\subseteq
V$, called \emph{bag}, such that
\begin{enumerate}[i)]
\item\label{treedec1} for each vertex~$v\in V$, there is a node~$x$
  of~$T$ with~$v\in V_x$,
\item\label{treedec2} for each edge~$\{u,w\}\in E$, there is a node~$x$
  of~$T$ with $\{u,w\}\subseteq V_x$,
\item\label{treedec3} for each vertex~$v\in V$, the nodes~$x$ of~$T$
  for which~$v\in V_x$ induce a subtree in~$T$.
\end{enumerate}
The \emph{width} of a tree decomposition is the size of
its largest bag minus one. The \emph{treewidth} of~$G$ is the minimum
width over all tree decompositions of~$G$.  The notions of \emph{path
  decomposition} and \emph{pathwidth} of~$G$ are defined in the same
way, except that $T$~is restricted to be a path.  The \emph{\iw{}} of
a hypergraph is the treewidth of its incidence graph.

\paragraph{Graph and hypergraph representations.}
When speaking about linear-time solvability, it is crucial to agree on
the graph and hypergraph representations we expect as input. We assume
that graphs are represented as \emph{adjacency lists}, that is, as a
list of vertices, each being associated with a list of its neighbors.

We assume hypergraphs to be given as \emph{hyperedge lists}, that is,
as a list of hyperedges, each being a list of the vertices it
contains. Note that a hypergraph given as hyperedge list is
linear-time transformable into an adjacency list of its incidence
graph and vice versa. Moreover, a hyperedge list is computable in
linear time from a hypergraph given as incidence matrix.

\paragraph{Tree automata.} A \emph{(deterministic leaf-to-root
  finite-state) tree automaton} is a
quintuple~$(Q,\Sigma,\delta,q_0,F)$, where $Q$~is a finite set of
\emph{states}, $\Sigma$~is a finite \emph{alphabet}, $q_0\in Q$ is the
\emph{start state} and $F\subseteq Q$~is the set of \emph{final
  states}, and, finally $\delta\colon(\Sigma\times Q)\cup(\Sigma\times
Q\times Q)\to Q$~is the \emph{transition function}. The set of all
rooted binary trees with vertices labeled using letters from~$\Sigma$
is denoted by~$\Sigma^{**}$.  We assume that each child of a node of a
rooted tree~$T\in\Sigma^{**}$ is fixed to be either a ``left'' or a
``right'' child.

A tree automaton processes a tree~$T\in\Sigma^{**}$ starting at its
leaves to determine the state at the root node of~$T$ as follows: the
state at a leaf node~$x$ of~$T$ with label~$a\in\Sigma$ is determined
by~$\delta(a,q_0)$. The state at a node~$x$ of~$T$ with label~$a$ and
a single child node~$y$ is determined by~$\delta(a,q_y)$, where
$q_y\in Q$~is the state at~$y$.  The state at a node~$x$ of~$T$ with
label~$a$, a left child~$y$, and a right child~$z$ is determined
by~$\delta(a,q_y,q_z)$, where $q_y,q_z\in Q$ are the states of~$y$
and~$z$, respectively.

A tree automaton \emph{accepts} a tree~$T\in\Sigma^{**}$ if its
state at the root node of~$T$ is in~$F$. A tree automaton~$A$
\emph{recognizes} a tree language~$L\subseteq\Sigma^{**}$ if, for
every tree~$T\in\Sigma^{**}$, the automaton~$A$ accepts~$T$ if and
only if~$T\in L$.

\looseness=-1 Note that an ordinary finite automaton for
\emph{words}~$w\in\Sigma^*$ over the alphabet~$\Sigma$ can be
understood as a tree automaton on rooted unary trees (paths).

\section{Myhill-Nerode for hypergraphs}\label{sec:mnhg}

\noindent The aim of this section is to generalize the Myhill-Nerode
theorem from formal languages to hypergraphs. To this end, we first
briefly recall the Myhill-Nerode theorem for formal languages in
\autoref{sec:languages}.

\autoref{sec:hygra} will prove the Myhill-Nerode theorem for
hypergraphs. Before, \autoref{sec:golgra} will generalize the
Myhill-Nerode theorem for graphs~\citep[Section~12.7]{DF13} to
vertex-colored graphs, since the Myhill-Nerode theorem for hypergraphs
will exploit that every hypergraph can be represented as its incidence
graph with two vertex types (or ``colors''): one type representing
hyperedges and one type representing the vertices of a hypergraph.

In \autoref{sec:fpt-mso}, we finally describe how our Myhill-Nerode
theorem yields linear-time algorithms for hypergraph problems and its
relation to the expressibility of hypergraph properties in monadic
second-order logic.

\subsection{Formal languages}\label{sec:languages}

The Myhill-Nerode theorem is a tool for proving or disproving that a
formal language is \emph{regular}, that is, decidable by a finite
automaton.  The theorem states that a language is regular if and only
if its so-called \emph{\crc} has a finite number of equivalence
classes.

\begin{definition}\label{def:crclang}
  Let $L\subseteq\Sigma^*$ be a language. The \emph{\crc}~$\sim_L$ is
  defined as follows: for~$v,w\in\Sigma^*,
  v\sim_L w:\iff\forall x\in\Sigma^*:vx\in L\iff wx\in L$,
  where $vx$~is the concatenation of~$v$ and~$x$. 
\end{definition}

\begin{example}
  Consider the language $L:=\{a^ib^j\mid i,j\in\mathbb
  N\}\subseteq\{a,b\}^*$ consisting of words starting with an
  arbitrary number of~$a$'s and ending in an arbitrary number
  of~$b$'s. Then, $a\sim_Laa$. However, $a\not\sim_Lab$, since, for
  example, $aa\in L$ but $aba\notin L$.
\end{example}

\noindent Obviously, for a language~$L\in\Sigma^*$, the
\crc{}~$\sim_L$ is an equivalence relation, that is, it is reflexive,
symmetric, and transitive. The \emph{index} of an equivalence relation
is the number of its equivalence classes.

\begin{mntheorem}
  A language~$L\subseteq\Sigma^*$ is recognizable by a finite
  automaton if and only if the \crc{}~$\sim_L$ has finite index.
\end{mntheorem}

\noindent Thus, the Myhill-Nerode theorem gives a necessary and
sufficient condition for a language being recognizable by a
finite automaton.

\subsection{Colored graphs}\label{sec:golgra}
\looseness=-1 In order to show the Myhill-Nerode theorem for hypergraphs, we first
present it for vertex-colored graphs. %
\citet[Section~12.7]{DF13} already proved the Myhill-Nerode
theorem for graphs without colors; \autoref{fig:mnexplain} gives a
rough overview of the technique. We will see that lifting it to
vertex-colored graphs is straightforward. Indeed,
\citet[Section~4.2.2]{CE12} provide an even more general Myhill-Nerode
theorem for graphs with vertex colors as well as edge colors. For our
purposes, however, a vertex-colored variant is sufficient and we show
it here as an introduction into the necessary concepts towards a
Myhill-Nerode theorem for hypergraphs.

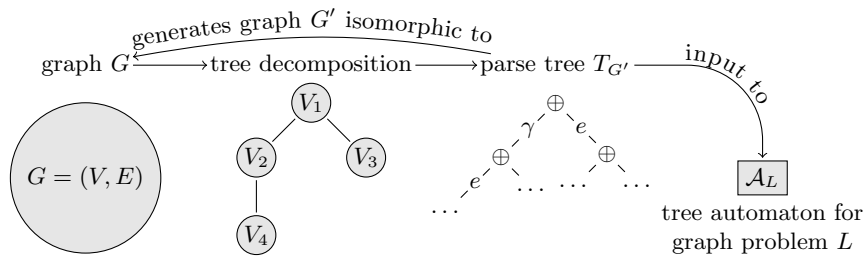
\begin{figure}
  \centering\footnotesize
  \def\myshift#1{\raisebox{1ex}}
  \def\myoshift#1{\raisebox{0.5ex}}
  \begin{tikzpicture}[auto, shorten >=1pt, shorten <=1pt]
    \tikzstyle{parse}=[circle, inner sep=0pt, minimum size=0mm]
    \tikzstyle{every node}=[inner sep=0pt]
    \node (graph) at (0,1.5) {graph~$G$};
    \node[red,minimum size=2cm] (v1) at (0,0) {$G=(V,E)$};

    \node (treedec) at (3,1.5) {tree decomposition};
    \draw [->] (graph) -> (treedec);

    \begin{scope}[shift={(3,1)}]
      \node[red] (v1) at (0,0) {$V_1$}; 
      \node[red,below right=5mm of v1] (v2) {$V_3$};
      \node[red,below left=5mm of v1] (v3) {$V_2$};
      \node[red,below=5mm of v3] (v4) {$V_4$};
      \draw (v1)--(v2); \draw (v1)--(v3); \draw (v3)--(v4);
    \end{scope}

    \node (parsetree) at (6.25,1.5) {parse tree~$T_{G'}$}; 
    \draw [->] (treedec) -> (parsetree);
    \begin{scope}[shift={(6.25,1)}]
      \node[parse] (v1) at (0,0) {$\oplus$};
      \node[parse,below      right=2mm of v1] (v2) {$e$};
      \node[parse,below right=2mm of      v2] (v3) {$\oplus$};
      \node[parse,below left=2mm of v3] (v4)      {$\dots$};
      \node[parse,below right=2mm of v3] (v8) {$\dots$};
      \node[parse,below left=2mm of v1] (v5) {$\gamma$};
      \node[parse,below left=2mm of v5] (v6) {$\oplus$};
      \node[parse,      below left=2mm of v6] (v9) {$e$};
      \node[parse, below left=2mm of      v9] (v11) {$\dots$};
      \node[parse, below right=2mm of v6] (v12)  {$\dots$};

      \draw (v1) -- (v2);
      \draw (v2) -- (v3);
      \draw (v3) -- (v4);

      \draw (v1) -- (v5);
      \draw (v5) -- (v6);

      \draw (v3) -- (v8);
      \draw (v6) -- (v9);
      \draw (v6) -- (v12);

      \draw (v9) -- (v11);
    \end{scope}
    \draw [<-,postaction={decorate,decoration={text along
        path,text align=center,text={|\myshift|generates
          graph~{\it $G'$} isomorphic to}}}] (graph) to[out=10,in=170]
    (parsetree);

    \begin{scope}[shift={(9,-0.5)}]
      \node[rectangle, inner sep=1mm, draw=black,fill=gray!20] (a) at (0,0.5)
      {$\mathcal A_L$}; 
      \draw[->,postaction={decorate,decoration={text along
          path,text align=center,text={|\myoshift|input to}}}] (parsetree) --(-1,2) to[out=0,in=90] (0,1) -- (a);
      \node[below=1mm of a, align=center] {tree automaton for\\graph problem~$L$};
    \end{scope}
  \end{tikzpicture}
  \caption{Solving a graph problem~$L$ using a tree automaton: from a
    graph~$G$ with bounded treewidth, a minimum width tree decomposition can be
    computed in linear time~\cite{Bod96}. The tree decomposition can
    be turned into a size-$O(n)$ expression over a
    set~$\{\emptyset,e,u,\gamma,i,\oplus\}$ of operators in linear
    time such that the value of the expression is a graph~$G'$
    isomorphic to~$G$~\cite[Theorem~12.7.1]{DF13}. The parse tree or
    expression tree~$T_{G'}$ of the expression is fed to a tree
    automaton~$\mathcal A_L$ that accepts~$T_{G'}$ in $O(n)$~time if
    and only if $G'\in L$. The existence of~$\mathcal A_L$ for the
    problem~$L$ can be proven or disproven by the Myhill-Nerode
    theorem for graphs.}
  \label{fig:mnexplain}
\end{figure}

\begin{figure}
 \centering
 \begin{tikzpicture}[shorten >=1pt, shorten <=1pt,node distance=0.5cm]
   \tikzstyle{parse}=[circle, inner sep=0pt, minimum size=0mm]

   \node[parse] (times) at (0,0) {$\cdot$};
   \node[parse,below right=of times] (plus1) {$+$};
   \node[parse,below right=of plus1] (2) {$2$};
   \node[parse,below left=of plus1] (3) {$3$};
   \node[parse,below left=of times] (sqrt) {$\sqrt{\cdot}$};
   \node[parse,below left=of sqrt] (plus2) {$+$};
   \node[parse,below right=of plus2] (5) {$4$};
   \node[parse,below left=of plus2] (4) {$5$};

   \draw (times)--(sqrt);
   \draw (sqrt)--(plus2);
   \draw (plus2)--(5);
   \draw (plus2)--(4);
   \draw (plus1)--(3);
   \draw (plus1)--(2);
   \draw (times)--(plus1);
 \end{tikzpicture}
 \caption{The expression tree or parse tree of the arithmetic expression~$\sqrt{5+4}\cdot(3+2)$.}
 \label{fig:parsetreeexplain}
\end{figure}
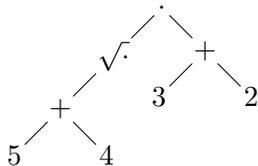

\bigskip\noindent In order to apply tree automata to graphs, we first
show how every graph of bounded treewidth can be represented by an
expression over a constant-size set of operators, and, consequently,
as the parse tree or expression tree of that expression
(\autoref{fig:parsetreeexplain} gives an example for a parse
tree). Herein, the crucial operator corresponds to the concatenation
of words in the language setting of the Myhill-Nerode theorem: like
every word with more than one letter is the concatenation of shorter
words, we will see that every graph of treewidth~$t-1$ with more than
$t$~vertices is isomorphic to the result of \emph{gluing} smaller
graphs together at a \emph{boundary} consisting of
$t$~vertices.\footnote{For keeping the presentation simpler, we leave
  graphs with less than $t$~vertices out of consideration: since we
  only consider constant values of~$t$ throughout this work, any
  problem restricted to graphs with less than $t$~vertices is a finite
  problem and can therefore be trivially solved by a finite
  automaton.}  This is formalized by the definition below and
illustrated in \autoref{bpjoinfig}.

\begin{definition}\label{colorglue}
  A \emph{\bndg{t}} $G$ is a graph with $t$~distinguished vertices
  that are labeled from $1$ to $t$. These labeled vertices are called
  \emph{boundary vertices}. The \emph{boundary}~$\partial(G)$ is the
  set of boundary vertices of~$G$.

  Two colored $t$-boundaried graphs $G_1$ and $G_2$ are
  \emph{isomorphic} and we write $G_1 \cong G_2$ if there is an
  isomorphism for the underlying (uncolored and unlabeled) graphs
  mapping each vertex to a vertex with the same color (but ignoring
  labels).

  Let~$G_1$ and~$G_2$ be \bndg{t}s whose vertices are colored with
  colors in~$\{1,\dots,\cmax\}$. We say that $G_1$~and~$G_2$ are
  \emph{color-compatible} if the vertices with the same labels
  in~$\partial(G_1)$ and~$\partial(G_2)$ have the same color.

  For two color-compatible \bndg{t}s, we denote by $G_1\oplusc G_2$
  the colored graph obtained by \emph{gluing}~$G_1$ and~$G_2$, that
  is, by taking the disjoint union of $G_1$~and~$G_2$ and identifying
  vertices of~$\partial(G_1)$ and~$\partial(G_2)$ having the same
  label; the vertex colors of~$G_1\oplusc G_2$ are inherited
  from~$G_1$ and~$G_2$.
  For two color-incompatible graphs~$G_1$ and~$G_2$, we leave
  $\oplusc$ undefined.
\end{definition}

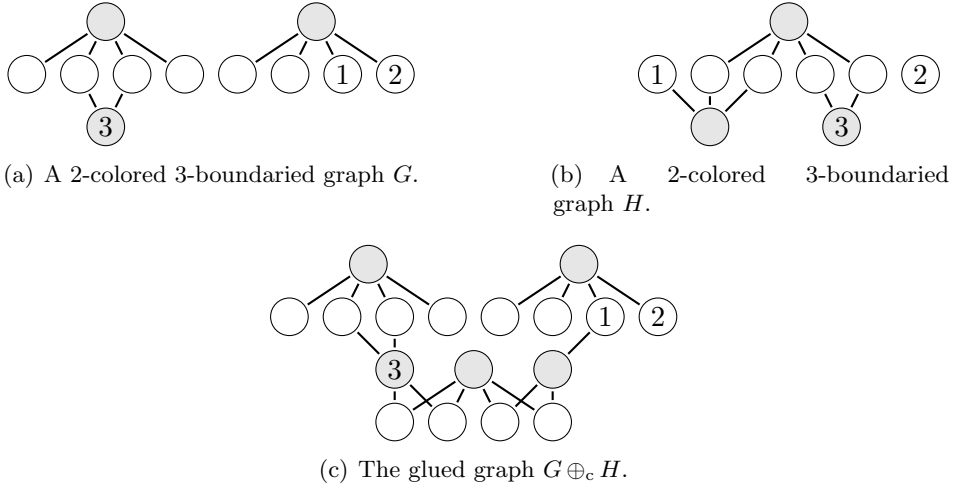
\begin{figure}[t]
 \centering
 \subfigure[A $2$-colored 3-boundaried graph~$G$.]{
   \begin{tikzpicture}[x=0.7cm,y=0.7cm,on grid, auto, shorten >=1pt, shorten <=1pt]
   
     \foreach \l in {1,...,8} {
       \node [blue] (v\l) at (\l,0) {\ifnum\l=7 {1}\else\ifnum\l=8 {2}\else{}\fi\fi};
     }
     
     \node [red] (v9) at (2.5,-1) {3};
     \node [red] (v10) at (6.5,1) {};
     \node [red] (v11) at (2.5,1) {};
     
     \foreach \s/\e in {9/2,9/3,10/5,10/6,10/7,10/8,11/1,11/2,11/3,11/4} {
       \draw [thick] (v\s) -- (v\e);
     }
   \end{tikzpicture}
 }\hfill{} \subfigure[A $2$-colored 3-boundaried graph~$H$.]{
   \hspace{1cm}\begin{tikzpicture}[x=0.7cm,y=0.7cm,on grid, auto, shorten >=1pt, shorten <=1pt]
     \foreach \l in {1,...,5} {
       \node [blue] (v\l) at (\l,0) {\ifnum\l=1 {1}\fi};
     }
     \node [blue] (v8) at (6,0) {2};
     
     \node [red] (v6) at (4.5,-1) {3};
     \node [red] (v7) at (2,-1) {};
     \node [red] (v9) at (3.5,1) {};
     
     \foreach \s/\e in {6/4,6/5,7/1,7/2,7/3,9/2,9/3,9/4,9/5} {
       \draw [thick] (v\s) -- (v\e);
     }
   \end{tikzpicture}
 }

  \subfigure[The glued graph $G\oplusc H$.]{
   \begin{tikzpicture}[x=0.7cm,y=0.7cm,on grid, auto, shorten >=1pt, shorten <=1pt]
   
     \foreach \l in {1,...,8} {
       \node [blue] (v\l) at (\l,0) {\ifnum\l=7 {1}\else\ifnum\l=8 {2}\else{}\fi\fi};
     }

     \foreach \l in {9,...,12} {
       \pgfmathtruncatemacro{\j}{\l-6}
       \node [blue] (v\l) at (\j,-2) {};
     }
     
     \node [red] (w1) at (6,-1) {};
     \node [red] (w2) at (6.5,1) {};
     \node [red] (w3) at (2.5,1) {};
     \node [red] (w4) at (3,-1) {3};
     \node [red] (w5) at (4.5,-1) {};
     
     \foreach \s/\e in {%
       2/5,2/6,2/7,2/8,
       1/7,1/11,1/12,
       3/1,3/2,3/3,3/4,
       4/2,4/3,4/9,4/10,
       5/9,5/10,5/11,5/12} {
       \draw [thick] (w\s) -- (v\e);
     }
   \end{tikzpicture}
 }
 \caption{Two color-compatible
   3-boundaried graphs~$G$ and~$H$ and their glued graph, where the boundary vertices are
   marked by their label.}
\label{bpjoinfig}
\end{figure}

\noindent Together with~$\oplusc$, we use the following set of
operators to create primitive graphs that can be glued together to
larger graphs using~$\oplusc$, and to arbitrarily permute the labels
on the boundary vertices.

\begin{definition}\label{ourops}The \emph{size-$t$ parsing operators}
  for $\{1, \dots$, $\cmax\}$-colored $t$-boundaried graphs are
  defined as follows:
  \begin{enumerate}[i)]
  \item $\{\emptyset_{n_1,\dots,n_{\cmax}} \mid
    \sum_{i=1}^{\cmax}n_i=t\}$ is a family of nullary operators that
    creates a graph consisting of isolated boundary
    vertices~$1,\dots,t$, of which the first $n_1$~vertices get
    color~$1$, the next $n_2$~vertices get color~$2$, and so on.
  \item $e$ is a unary operator that adds an edge between the boundary
    vertices labeled~1 and~2.
  \item $\{u_\ell : 1\leq\ell\leq\cmax\}$ is a family of unary
    operators that add a new boundary vertex of color~$\ell$ and
    labels it~1, unlabeling the vertex previously labeled~1.
  \item $\gamma$ is a unary operator that cyclically shifts the
    boundary. That is, $\gamma$~moves label~$j$ to the vertex with
    label~$j+1\pmod{t}$.
  \item $i$ is a unary operator that assigns the label~1 to the vertex
    currently labeled 2 and label~2 to the vertex with label~1.
  \item $\oplusc$ is our gluing operator from \autoref{colorglue}.
  \end{enumerate}
\end{definition}

\noindent For a constant number of colors~$\cmax$, the set of size-$t$
parsing operators is finite. Moreover, for~$\cmax=1$, the given
operators coincide with those given by \citet[Section~12.7]{DF13}
for uncolored graphs, which allows us to show the following theorem:

\begin{theorem}\label{thm:parsethm}
  Let $G$~be a \{1,\dots,$\cmax$\}-colored graph with constant
  treewidth $t-1$ and at least $t$~vertices.

  Then, in linear time, $G$~can be transformed into an expression over
  the size-$t$ operators in \autoref{ourops} whose value is a
  graph~$H$ that is isomorphic to~$G$, that is, when ignoring the labels of~$H$.
\end{theorem}

\begin{proof}
  \looseness=-1 For $\cmax=1$, \citet[Theorem~12.7.1]{DF13} provide a
  linear-time procedure for converting a tree decomposition of
  width~$t-1$ of a graph~$G$ into an expression over the size-$t$
  operators in \autoref{ourops} such that the value of the expression
  is a graph isomorphic to~$G$. This procedure is easily adapted for
  larger~$\cmax$: whenever \citet{DF13} introduce vertices
  using~$\emptyset_{1}$ or~$u_1$ in the case~$\cmax=1$, we introduce
  them using~$\emptyset_{n_1,\dots,\cmax}$ and~$u_\ell$ with the
  colors they have in~$G$.
\end{proof}

\noindent We are now at a point where we can get each graph of
bounded treewidth into a representation that we can feed into a tree
automaton: we use the parse tree (or expression tree) of an expression
over the operators in \autoref{ourops}. A central question remains:
which graph problems can be decided by a tree automaton operating on
such a parse tree? The Myhill-Nerode theorem for colored graphs will
give a sufficient and necessary condition. To state the theorem, we
first lift the concept of a \crc{} from the language setting
(\autoref{def:crclang}) to graphs.

\begin{definition}\label{def:ulargesmall}
  Let $\ularge{\cmax}$ be the \emph{large universe} of all
  $\{1,\dots,\cmax\}$-colored \bndg{t}s and
  $\usmall{\cmax}\subseteq\ularge{\cmax}$ be the \emph{small universe} of
  $\{1,\dots, \allowbreak c_{\max}\}$-colored \bndg{t}s that can be
  generated by the size-$t$ operators in \autoref{ourops}.

  For $U\in\{\usmall{\cmax},\ularge{\cmax}\}$, we say that $F\subseteq
  U$ is a \emph{graph problem} if, for all~$G\in F$ and $H\in U$ with
  $G\cong H$, we also have $H\in F$.  That is, we assume graph
  problems to be closed under isomorphism and, in particular, that
  changing vertex labels does not influence membership in~$F$.
  
  Finally, for a graph problem~$F\subseteq U$, where
  $U\in\{\usmall{\cmax},\ularge{\cmax}\}$, we define the
  \emph{\crc{}~$\sim_F$ over~$U$} for~$F$ as follows:
  for~$G_1,G_2\in U$, $G_1\sim_F G_2$ if and only if
  $G_1$~and~$G_2$ are color-compatible and if for all
  color-compatible~$H\in U$, we have $G_1\oplusc H\in
  F\iff G_2\oplusc H\in F$.
\end{definition}

\looseness=-1 \noindent \autoref{thm:parsethm} might mislead to the
impression that $\usmall{\cmax}$~contains all
treewidth-$(t-1)$ graphs of~$\ularge{\cmax}$.  However, this is not
the case: for example, a path of length more than one whose first
vertex has label~1 and whose last vertex has label~2 has treewidth~1
and is contained in $\ularget{1}$.  However, it cannot be generated by
the size-2 operators in \autoref{ourops} and, hence, is not in
$\usmallt{1}$.  We will see this detail to be important when showing
that a graph problem is \emph{not} recognizable by a finite tree
automaton.

\begin{theorem}\label{afclemma}
  Let $F\subseteq \usmall{\cmax}$~be a graph problem. The following
  statements are equivalent:
  \begin{enumerate}[i)] %
  \item The collection of parse trees generating the graphs
    in~$F$ is recognizable by a finite tree
    automaton.
  \item The \crc{} $\sim_{F}$ has finite index over~$\usmall{\cmax}$.
  \end{enumerate} %
\end{theorem}

\noindent\looseness=-1 \citet[Theorem~12.7.2]{DF13} proved \autoref{afclemma} for uncolored
graphs, that is, for~$\cmax=1$.  The case~$\cmax>1$ can be proven
analogously.  However, the proof of (ii)${}\rightarrow{}$(i) given by
\citet{DF13} is flawed and we correct it in \autoref{apxfix}.

\subsection{Hypergraphs}
\label{sec:hypergraphs}\label{sec:hygra}

\noindent In this section, we show how tree automata can be used to
recognize hypergraph properties and, in the form of a Myhill-Nerode
theorem for hypergraphs, a necessary and sufficient characterization
for the hypergraph properties that a tree automaton can decide. To
this end, we first define the notion of gluing for hypergraphs.

\begin{definition}\label{hgglue}
  A \bndh{t}~$H$ has $t$~distinguished vertices and hyperedges labeled
  from~$1$ to~$t$ called \emph{boundary objects}. The
  \emph{boundary}~$\partial(H)$ is the set of all boundary
    objects.

  Two \bndh{t}s are \emph{gluable} if no vertex of one hypergraph has
  the label of a hyperedge of the other hypergraph.

  Let $H_1$~and~$H_2$ be two gluable \bndh{t}s. We denote
  by~$H_1\oplush H_2$ the \bndh{t} obtained by taking the disjoint
  union of~$H_1$ and~$H_2$, identifying each labeled vertex of~$H_1$
  with the vertex of~$H_2$ with the same label, and replacing the
  hyperedges with the same label~$\ell$ by their union.
\end{definition}

\noindent In order to apply tree automata to hypergraphs, in contrast to
\autoref{sec:golgra} for colored graphs, we will not define a set of
additional operators for generating hypergraphs. Instead, we will
generate hypergraphs from two-colored incidence graphs: vertices of
one color will represent the vertices of the hypergraph, vertices of
the other color will represent the hyperedges. That is, instead of
solving a hypergraph problem, we will in fact solve a graph problem on
colored incidence graphs. The goal of the next definition is to give a
representation of a hypergraph problem as a graph problem. It is
illustrated in \autoref{hggluefig}.

\begin{figure}[t]\centering
  \subfigure[The 3-boundaried hypergraph~$\mathcal H(G)$.]{%
    \begin{tikzpicture}[x=0.8cm, y=0.55cm, on grid, auto, shorten >=1pt, shorten <=1pt]
      \tikzstyle{edge} = [color=black,opacity=.2,line cap=round, line join=round, line width=19pt]    

      \foreach \l in {1,...,8} {
        \node [blue] (v\l) at (\l,0) {\ifnum\l=7 {1}\else\ifnum\l=8 {2}\else{}\fi\fi};
      }

      \node (l) at (2.5,1) {3};
      \begin{pgfonlayer}{background}
        \draw [edge] (v1.center) -- (v4.center);
        \draw [edge] (v5.center) -- (v8.center);
        \draw [edge, line width=23pt] (v2.center) -- (v3.center) -- (2.5,1) -- cycle;
      \end{pgfonlayer}
    \end{tikzpicture}
  }\hfill{}
  \subfigure[The 3-boundaried hypergraph~$\mathcal H(H)$.]{
    \begin{tikzpicture}[x=0.8cm, y=0.55cm,on grid, auto, shorten >=1pt, shorten <=1pt]
      \tikzstyle{edge} = [color=black,opacity=.2,line cap=round, line join=round, line width=19pt]    

      \foreach \l in {1,...,6} {
        \node [blue] (v\l) at (\l,0) {\ifnum\l=1 {1}\else\ifnum\l=6 {2}\else{}\fi\fi};
      }

      \node (l) at (4.5,1) {3};
      \begin{pgfonlayer}{background}
        \draw [edge] (v2.center) -- (v5.center);
        \draw [edge, line width=22pt] (v1.center) -- (v3.center);
        \draw [edge, line width=22pt] (v4.center) -- (v5.center) -- (4.5,1) -- cycle;
      \end{pgfonlayer}
    \end{tikzpicture}
  }
  \subfigure[The glued hypergraph~$\mathcal H(G)\oplush\mathcal H(H)=\mathcal H(G\oplusc H)$.]{%
    \begin{tikzpicture}[shorten >=1pt, shorten <=1pt]
      \tikzstyle{edge} = [color=black,opacity=.2,line cap=round, line join=round, line width=19pt]    

      \foreach \l in {1,...,12} {
        \node [blue] (v\l) at (\l,0) {\ifnum\l=7 {1}\else\ifnum\l=8 {2}\else{}\fi\fi};
      }

      \node (l) at (6,1) {3};
      \begin{pgfonlayer}{background}
        \draw [edge] (v1.center) -- (v4.center);
        \draw [edge] (v5.center) -- (v8.center);
        \draw [edge] (v9.center) -- (v12.center);
        \draw [edge, line width=25pt] (3,1) -- (v3.center) -- (v2.center) -- (3,1) -- (11,1) -- (v11.center) -- (v12.center) -- (11,1);

        \draw [edge, line width=22pt] (v7.center) -- (7,1) -- (9,1) -- (v9.center) -- (v10.center) -- (9,1);
      \end{pgfonlayer}
    \end{tikzpicture}
  }
  \caption{The two hypergraphs represented by the \wcg{3}s~$G$ and~$H$
    in \autoref{bpjoinfig} and the glued hypergraph
    $\hyperg{G}\oplush\hyperg{H}=\hyperg{G\oplusc H}$.}
  \label{hggluefig}
\end{figure}

\begin{definition}\label{liftbip}
  \looseness=-1 A \emph{\wcg{t}} is a $\{1,2\}$-colored \bndg{t}
  $G=(U\uplus W,E)$ such that all vertices in $U$~have color~1 and all
  vertices in $W$~have color 2, and each of $U$~and~$W$ form an
  independent set.

  For a \wcg{t}~$G=(U\uplus W,E)$, we denote by $\hyperg{G}$ the \bndh
  t with the vertex set~$U$ and the hyperedge set~$\{N(w) \mid w\in
  W\}$. Moreover, each vertex of~$\hyperg{G}$ inherits its label
  from~$G$ and each hyperedge~$e$ in~$\hyperg{G}$ inherits its label
  from the vertex~$w\in W$ of~$G$ that induced~$e$.

  For a set~$F\subseteq\usmall{2}$ of \wcg{t}s, we denote
  $\hyperg{F}:=\bigcup_{G\in F}\hyperg{G}$ and we call
  $F$~\emph{\full{}} if, for all \wcg{t}s~$G\in\usmall{2}$,
  $\hyperg{G}\in\hyperg{F}\implies G\in F$.

  We use~$\hlarge$ to denote the \emph{large universe} of all \bndh ts
  and by~$\hsmall$ we denote the \emph{small
    universe}~$\hyperg{\usmall{2}}$, that is, the \bndh{t}s that can
  be generated from \wcg{t}s created by the operators in
  \autoref{ourops}.

  We say that $\mathcal F\subseteq U$ for $U\in\{\hlarge,\hsmall\}$ is
  a \emph{hypergraph problem} if, for all~$G\in \mathcal F$ and $H\in
  U$ with $G\cong H$, we also have $H\in \mathcal F$.  That is, we
  assume hypergraph problems to be closed under isomorphism and, in
  particular, that changing boundary labels does not influence
  membership in~$\mathcal F$.
\end{definition}

\noindent The following observation allows us, where helpful, to
denote hypergraphs~$H$ using~$\hyperg{G}$ for some graph~$G$
with~$\hyperg{G}=H$, and to denote hypergraph problems~$\mathcal F$
using~$\hyperg{F}$ for some \full{}
graph problem~$F$.

\begin{observation}\label{obspscomplete}\leavevmode  
  \begin{enumerate}[i)]
  \item A graph~$G\in\ularge{2}$ is isomorphic to the incidence graph
    of~$\hyperg{G}\in\hlarge$. Therefore, the treewidth of~$G$ equals
    the incidence treewidth of~$\hyperg{G}$.

  \item For two graphs~$G,H\in\ularge{2}$, we have
    $\hyperg{G}\oplush\hyperg{H}=\hyperg{G\oplusc H}$.
    
  \item For a \full{} $F\subseteq\usmall{2}$ and each \wcg{t}~$G\in
    \usmall{2}$, we have $G\in F$ if and only
    if $\hyperg{G}\in\hyperg{F}$.\footnote{\looseness=-1If $F$~is not \full{}, it
      might be that~$G\notin F$ but $\hyperg{G}\in\hyperg{F}$
      because~$H\in F$ for some \wcg{t}~$H\neq G$ with
      $\hyperg{G}=\hyperg{H}$: the graphs~$G$ and~$H$ might represent
      the hyperedges of~$\hyperg{G}=\hyperg{H}$ using different
      mathematical objects.}
    
  \item For every \bndh t~$H\in\hsmall$, by definition of~$\hsmall$,
    there is a \bndg t~$G\in\usmall{2}$ such that~$\hyperg{G}=H$.
    Consequently, for every hypergraph
    problem~$\mcF{}\subseteq\hsmall$, there is a \full{}
    $F\subseteq\usmall{2}$ with $\hyperg{F}=\mcF{}$. Moreover, in
    terms of \autoref{def:ulargesmall}, $F$~is a graph problem.
  \end{enumerate}
\end{observation}

\noindent In order to state the Myhill-Nerode theorem for hypergraphs,
we define the \crc{} for hypergraphs.

\begin{definition}
  Let $\mcF{}\subseteq U$ for $U\in\{\hlarge,\hsmall\}$ be a hypergraph problem. We define the
  \emph{\crc{}~$\sim_{\mcF{}}$ over~$U$} for~$\mcF{}$ as follows:
  for $G_1,G_2\in U$, $G_1\sim_{\mathcal F} G_2$ if and only if
  $G_1$ and $G_2$ are gluable and for all $H\in U$ that are
  gluable to~$G_1$ and~$G_2$, $G_1\oplush H\in \mathcal F\iff
  G_2\oplush H\in \mathcal F$.
\end{definition}

\noindent We now state our Myhill-Nerode theorem for hypergraphs. As
\autoref{afclemma}, the following theorem only makes a statement
about when a tree automaton can decide hypergraph problems~$\mathcal
F\subseteq \hsmall$. However, in \autoref{sec:fpt-mso}, we will see that this
restriction is not important in most cases.

\begin{theorem}\label{afhlemma}
  Let $\mcF{}\subseteq\hsmall$ be a hypergraph problem, that is,
  $\mcF{}=\hyperg{F}$ for some \full{} $F\subseteq\usmall{2}$.  The
  following statements are equivalent:
  \begin{enumerate}[i)]
  \item The collection of parse trees generating the graphs in~$F$
    is recognizable by a tree automaton.
  \item The \crc{} $\sim_{\mcF{}}$ has finite index
    over~$\hsmall$.
  \item The \crc{} $\sim_{F}$ has finite index over~$\usmall{2}$.
  \end{enumerate}
  Moreover, if the index~$p$ of~$\sim_\mcF{}$ and the index~$q$
  of~$\sim_F$ are finite, they bound each other as $2^tq\geq p\geq
  q/2^t-1$.
\end{theorem}
\begin{proof}%
  Since $F\subseteq\usmall{2}$, we can apply \autoref{afclemma}, which
  states that (i) and (iii) are equivalent. It remains to show that
  (iii) and (ii) are equivalent. That is, we show that $\sim_{\mcF{}}$
  has finite index over~$\hsmall$ if and only if $\sim_{F}$ has finite
  index over~$\usmall{2}$.

  First, assume that $\sim_{F}$ has infinite index over~$\usmall{2}$.
  We show that $\sim_{\mcF{}}$ has infinite index
  over~$\hsmall$. Since $\sim_{F}$ has infinite index
  over~$\usmall{2}$, there is an infinite set~$\{G_1,G_2,$
  $G_3,\dots\}\subseteq\usmall{2}$ of graphs that are pairwise
  nonequivalent under~$\sim_{F}$. Since there are only
  $2^t$~possibilities to assign two colors to $t$~boundary vertices,
  there is an infinite number of color-compatible graphs
  among~$\{G_1,G_2,\dots\}$. Moreover, notice that all
  graphs~$G_i\in\usmall{2}$ that are not \wcg{t}s are equivalent
  under~$\sim_{F}$: since $F$~contains only \wcg{t}s, $G_i$~cannot be
  completed into graphs in~$F$ by gluing any graph onto~$G_i$.
  Therefore, without loss of generality, we assume that
  $\{G_1,G_2,\dots\}$ are pairwise color-compatible \wcg{t}s.  Now,
  for each pair~$G_i,G_j$, there is a graph~$H_{ij}\in\usmall{2}$ such
  that, without loss of generality, $G_i\oplusc H_{ij}\in F$ but
  $G_j\oplusc H_{ij}\notin F$.  From $G_i\oplusc H_{ij}\in F$, it
  follows that~$H_{ij}$ is a \wcg{t} that is color-compatible
  with~$G_i$. Hence, $\hyperg{H_{ij}}\in\hsmall$. Now, from
  $G_i\oplusc H_{ij}\in F$, we get
  $\hyperg{G_i}\oplush\hyperg{H_{ij}}=\hyperg{G_i\oplusc
    H_{ij}}\in\hyperg{F}=\mcF{}$.  Moreover, since $F$ is \full{},
  from $G_j\oplusc H_{ij}\notin F$ it follows that
  $\hyperg{G_j}\oplush \hyperg{H_{ij}}=\hyperg{G_j\oplusc
    H_{ij}}\notin \hyperg{F}=\mcF{}$. That is,
  $\hyperg{G_i}\nsim_{\mcF{}}\hyperg{G_j}$ over~$\hsmall$ and,
  therefore, $\sim_{\mcF{}}$ has infinite index.

  Now, assume that $\sim_{\mcF{}}$ has infinite index
  over~$\hsmall$. We show that $\sim_{F}$ has infinite index
  over~$\usmall{2}$. Since $\sim_{\mcF{}}$ has infinite index
  over~$\hsmall$, there is a
  set~$\{\hyperg{G_1},\allowbreak\hyperg{G_2},$
  $\hyperg{G_3},\dots\}\subseteq\hsmall$ of hypergraphs that are
  pairwise nonequivalent under~$\sim_{\mcF{}}$. Since there are only
  $2^t$~partitions of $t$~labels into hyperedge labels and vertex
  labels, there is an infinite number of pairwise gluable hypergraphs
  among~$\{\hyperg{G_1},\allowbreak\hyperg{G_2},\dots\}$. Therefore,
  without loss of generality, assume that all these hypergraphs are
  pairwise gluable. Now, for each pair $\hyperg{G_i}$, $\hyperg{G_j}$,
  there is a hypergraph~$\hyperg{H_{ij}}\in\hsmall$ such that, without
  loss of generality, we have $\hyperg{G_i}\oplush
  \hyperg{H_{ij}}\in\hyperg{F}=\mcF{}$ but
  $\hyperg{G_j}\oplush\allowbreak
  \hyperg{H_{ij}}\notin\hyperg{F}=\mcF{}$. Since
  $\hyperg{G_i}\oplush\allowbreak{}\hyperg{H_{ij}}=\allowbreak{}\hyperg{G_i\oplusc
    H_{ij}}$ and $F$~is \full{}, we have $G_i\oplusc H_{ij}\in
  F$. Moreover, $G_j\oplusc H_{ij}\notin F$. Since
  $H_{ij}\in\usmall{2}$, it follows that $\sim_{F}$ has infinite index
  over~$\usmall{2}$.

  Finally, observe that our proof yields even more
  information in the case that~$\sim_F$, and equivalently
  $\sim_\mcF{}$, have finite index: we have first shown that, for any
  set of~$q$ graphs that are pairwise nonequivalent under~$\sim_F$,
  there are at least $p\geq \lceil q/2^t\rceil-1$~hypergraphs
  nonequivalent under~$\sim_\mcF{}$.  We have then shown that, for any
  set of $p$~hypergraphs that are pairwise nonequivalent
  under~$\sim_\mcF{}$, there are at least $q\geq \lceil
  p/2^t\rceil$~graphs nonequivalent under~$\sim_F$.  Hence, for the
  index~$p$ of~$\sim_\mcF{}$ and the index~$q$ of~$\sim_F$, we
  have~$2^tq\geq p\geq q/2^t-1$.
\end{proof}

\subsection{Fixed-parameter algorithms and monadic second-order
  logic}\label{sec:fpt-mso}

In \autoref{sec:hygra}, we have seen a tool allowing us to show when a
hypergraph problem~$\mathcal F\in\hsmall$ can be recognized by a
finite tree automaton. The provided \autoref{afhlemma}, however, is
strongly tied to the representation of hypergraphs as incidence graphs
in~$\usmall{2}$. This section shows three corollaries to ease the
application of \autoref{afhlemma} for classifying hypergraph problems.

\paragraph{Showing tractability and constructing tree automata.} The
following corollary will make it easier to show that a hypergraph
problem is fixed-parameter linear parameterized by \iw{}. Essentially,
we do not have to care about whether the hypergraphs we consider are
contained in~$\hsmall$.

\begin{corollary}\label{hisoclosed}
  Let $\mcF{}\subseteq\hlarge$ be a decidable hypergraph problem and
  that is restricted to hypergraphs of constant \iw{}~$t-1$.

\nopagebreak
  Given a hypergraph~$H\in\hlarge{}$ and a constant upper bound on the
  index of $\sim_{\mcF{}}$ over~$\hlarge{}$, we can compute in
  constant time a tree automaton~$\mathcal A$ and in linear time a
  tree~$T$ such that $\mathcal A$ processes~$T$ in linear time and
  accepts~$T$ if and only if~$H\in\mcF{}$.
\end{corollary}

\begin{proof}
  Let $F\subseteq\usmall{2}$ be \full{} such
  that~$\mcF{}\cap\hsmall=\hyperg{F}$ and let $p$~be the constant
  given upper bound on $\sim_\mcF{}$.  The proof relies on two~claims:
  \begin{enumerate}[i)]
  \item $H\in \mcF{}$ holds if and only if all
    graphs~$G\in\usmall{2}$ with~$\hyperg{G}\cong H$ are in~$F$.
  \item $\sim_{F}$ has index~$q\leq2^t(p+1)$ over~$\usmall{2}$.
  \end{enumerate}
  From~(i) then immediately follows that a tree automaton~$\mathcal
  A$ deciding~$F$ decides $H\in \mcF{}$~correctly when fed the parse
  tree~$T_G$ of any~$G\in\usmall{2}$ with~$\hyperg{G}\cong H$. By
  \autoref{thm:parsethm}, this $G$~exists and we obtain the parse
  tree~$T_G$ in linear time from the incidence graph of~$H$.

  \looseness=-1 From~(ii) and \autoref{afhlemma}, it follows that the
  tree automaton $\mathcal A$ indeed exists. It can be constructed in
  constant time given that we know a constant upper bound~$s$ on the
  number of states of~$\mathcal A$: the
  crucial observation is that $\mathcal A$ reaches at least one state
  twice when processing a parse tree of height greater than~$s$.
  Thus, for any parse tree~$T$ of height greater than~$s$, there is a
  parse tree~$T'$ of height at most~$s$ such that $\mathcal A$ accepts
  $T$~if and only if it accepts~$T'$.  It follows that we only have to
  construct~$\mathcal A$ so that it works correctly on all parse trees
  of height at most~$s$.  Since our operators in \autoref{ourops} are
  all nullary, unary, or binary, the parse trees of expressions over
  them are binary trees.  Thus, there are only a constant number of
  parse trees of height at most~$s$.  Moreover, for any parse tree
  $T$~of height at most~$s$, we can decide in constant time whether
  the hypergraph it generates is in~$\mcF{}$, since $\mcF{}$~is
  decidable.  Hence, we can construct in constant time by brute force
  a tree automaton~$\mathcal A$ that correctly answers for all parse
  trees of height at most~$s$ and, consequently, recognizes~$F$.
  Moreover, since $\mathcal A$ has a constant number of states, it
  takes only linear time for $\mathcal A$ to process any parse tree.

  The constant upper bound on the number of states of~$\mathcal
  A$ we obtain as follows: the index of~$\sim_F$ is bounded by~(ii),
  which, in turn bounds the number of states of~$\mathcal
  A$~\citep[Theorem~12.7.2]{DF13}.  Thus, it only remains to prove~(i)
  and~(ii).

  \looseness=-1 (i) First, assume that there is some graph~$G\in\usmall{2}$
  with~$\hyperg{G}\cong H$ in~$F$.  Then, since~$G$ is a \wcg{t}
  in~$\usmall{2}$, $\hyperg{G}\in\hyperg{F}$. Since
  $\hyperg{F}\subseteq\mcF{}$ and~$\mcF{}$~is closed under
  isomorphism, we conclude $H\in\mcF{}$.  If, for the opposite
  direction, $H\in\mcF{}$, then let %
  $G\in\usmall{2}$ be any graph such that $\hyperg{G}\cong H$. Since
  $G\in\usmall{2}$, we have~$\hyperg{G}\in\hsmall$. Moreover, since
  $\mcF{}$~is closed under isomorphism, $\hyperg{G}\in \mcF{}$ and,
  hence, $\hyperg{G}\in\hyperg{F}$. Finally, since $F$~is \full{}, we
  have~$G\in F$.

  (ii) We show that the index~$p'$ of~$\sim_{\hyperg{F}}$
  over~$\hsmall{}$ is at most~$p$.  Then, from \autoref{afhlemma}, it
  follows that~$p\geq p'\geq q/2^tp-1$ and, hence, $q\leq2^t(p+1)$.
  Thus, we only have to show, for any
  $\hyperg{G_1},\hyperg{G_2}\in\hsmall{}$ equivalent
  under~$\sim_\mcF{}$, that they are also equivalent
  under~$\sim_{\hyperg{F}}$.  This is trivial, since, for
  $i\in\{1,2\}$ and any~$\hyperg{H}\in\hsmall{}$, we
  have~$\hyperg{G_i}\oplush\hyperg{H}\in\hyperg{F}\iff\hyperg{G_i}\oplush\hyperg{H}\in\mcF{}$,
  since $\hyperg{G_i}\oplush\hyperg{H}=\hyperg{G_i\oplusc
    H}\in\hsmall$ and $\hyperg{F}=\mcF\cap\hsmall{}$.
\end{proof}

\noindent From \autoref{hisoclosed}, it follows that, to obtain a
fixed-parameter linear algorithm for some hypergraph
problem~$\mathcal F$ parameterized by \iw{}, we just have to show
that~$\sim_{\mathcal F}$ has finite index over~$\hlarge$. 

\paragraph{Showing intractability.}

We have seen that it was enough to show that $\sim_{\mathcal F}$~has
finite index over~$\hlarge$ to show that a hypergraph
problem~$\mathcal F\subseteq\hlarge{}$ of hypergraphs with incidence
treewidth~$t-1$ is decidable by a tree automaton. To show the
opposite, it is not sufficient to show that~$\sim_{\mathcal F}$ has
infinite index over the hypergraphs with treewidth~$t-1$ in~$\hlarge$:
assume, for example, that there are two
hypergraphs~$H_i,H_j\in\hsmall$ that are non-equivalent
under~$\sim_{\mathcal F}$. Then, there is some~$H_{ij}\in\hlarge$
satisfying~$\hyperg{G_i}\oplush H_{ij}\in\mathcal F$
but~$\hyperg{G_j}\oplush H_{ij}\notin\mathcal F$. If
$H_{ij}\notin\hsmall{}$, then this does not necessarily mean
that~$H_i$ and~$H_j$ are nonequivalent under~$\sim_{\mathcal F}$
over~$\hsmall{}$. Thus, we cannot conclude that~$\sim_{\mathcal F}$
has infinite index over~$\hsmall{}$ and \autoref{afhlemma} is
inapplicable. However, the following corollary gives a simple
criterion for intractability.

\begin{corollary}\label{hisoclosed2}
  Let $\mathcal F\subseteq\hlarge$ be a hypergraph problem %
  and $\mathcal F_t\subseteq\mathcal F$~be an arbitrary
  subset of hypergraphs~$H$ whose incidence
  graphs have tree decompositions of width~$t-1$ such
  that~$\partial(H)$~is a bag.

  If $\sim_{\mathcal F_t}$ has infinite index over~$\hlarge$, then $\sim_F$
  has infinite index over~$\usmall{2}$ for $F\subseteq\usmall{2}$ being the
  \full{} graph problem such that~$\hyperg{F}=\mathcal F\cap\hsmall$.

  Consequently, there is no tree automaton that decides~$H\in\mathcal
  F$ correctly when fed the parse tree of the incidence graph of~$H$.
\end{corollary}

\begin{proof}
  Any hypergraph $H\in\mathcal F_t$ allows for a tree
  decomposition~$T$ of width~$t-1$ of its incidence graph that has a
  bag~$\partial(H)$. The procedure by \citet[Theorem~12.7.1]{DF13}
  produces a parse tree for a graph~$G\in\usmall{2}$
  with~$\hyperg{G}\cong H$. The crucial observation is that, when
  choosing the bag~$\partial(H)$ as the root of the tree
  decomposition~$T$, the procedure generates a parse tree for a
  graph~$G\in\usmall{2}$ with~$\hyperg{G}\cong_t H$, where we
  use~$\cong_t$ to denote that there is an isomorphism
  between~$\hyperg{G}$ and~$H$ that maps the $t$~boundary objects
  of~$\hyperg{G}$ to boundary objects in~$H$ with the same label.

  Now, let $\{H_1,H_2,\dots\}\subseteq\mathcal F_t$ be a set of
  hypergraphs that are pairwise non-equivalent with respect
  to~$\sim_{\mathcal F_t}$. As before, we may assume that they are
  pairwise gluable. Hence, for each pair~$H_i,H_j$, there is a
  hypergraph~$H_{ij}\in\hlarge$ such that~$H_i\oplush
  H_{ij}\in\mathcal F_t$ but~$H_j\oplush H_{ij}\notin\mathcal
  F_t$. Since $H_i\oplush H_{ij}\in\mathcal F_t$, the
  hypergraph~$H_i\oplush H_{ij}$ has a tree decomposition of
  width~$t-1$ with a bag~$\partial(H_i\oplush H_{ij})$ and, hence,
  $H_{ij}$ has such a tree decomposition as well.
  It follows that there are graphs~$G_i,G_j,G_{ij}\in\usmall{2}$ such
  that~$\hyperg{G_i}\cong_t H_i$, $\hyperg{G_j}\cong_t H_j$, and
  $\hyperg{G_{ij}}\cong_t H_{ij}$. Then, $\hyperg{G_i\oplusc G_{ij}}=
  \hyperg{G_i}\oplush\hyperg{G_{ij}}\cong_t H_i\oplush
  H_{ij}\in\mathcal F$. Since $\mathcal F$~is closed under
  isomorphism, it follows that~$\hyperg{G_i\oplusc G_{ij}}\in\mathcal
  F\cap\hsmall$. Since $F$~is \full{}, $G_i\oplusc G_{ij}\in F$. With
  the same argumentation, it follows that~$G_j\oplusc G_{ij}\notin
  F$. 

  It follows that $\sim_F$ has infinite index over~$\usmall{2}$ and by
  \autoref{afhlemma}, there is no tree automaton recognizing the parse
  trees in~$F$.
\end{proof}

\paragraph{Excluding expressibility in monadic second-order logic.}
A standard way of showing linear-time solvability of a graph
problem~$F$ on graphs of bounded treewidth is expressing the property
of being a yes-instance of~$F$ in monadic second-order logic of
graphs~\citep[Section~10.6]{Nie06}.

Previously, we have seen how to show that some
hypergraph problem~$\mathcal F$ cannot be solved by a finite tree
automaton. An immediate consequence is that the property of being a
yes-instance for~$\mathcal F$ is not expressible in monadic
second-order logic for hypergraphs; we now give a little detail about
this connection.

\begin{definition}[Monadic second-order logic for graphs and
  hypergraphs] A formula in the \emph{MS$_1$-logic for
    $\{1,\dots,\cmax\}$-colored graphs} may consist of the logic
  operators~$\vee,\wedge,\neg$, vertex variables, set variables,
  quantifiers~$\exists$ and~$\forall$ over vertices and vertex sets,
  and the predicates
  \begin{enumerate}[i)]
  \item $x\in X$ for a vertex variable~$x$ and a set~$X$,
  \item $\adj(v,w)$, being true if $v$~and~$w$ are adjacent vertices,
  \item $\lab_i(v)$ for $1\leq i\leq\cmax$, being true if $v$~is a
    vertex with color~$i$,
  \item equality of vertex variables and set
    variables.
  \end{enumerate}
  A formula of the \emph{MS$_2$-logic for hypergraphs} may
  consist of the logic operators~$\vee,\wedge,\neg$, vertex variables,
  hyperedge variables, set variables, quantifiers~$\exists$ and~$\forall$
  over vertices, hyperedges, and sets, and the predicates
  \begin{enumerate}[i)]
  \item$x\in X$ for a vertex or hyperedge variable~$x$ and a set~$X$,
  \item$\inc(e,v)$, being true if $e$~is a hyperedge containing~$v$,
  \item$\adj(v,w)$, being true if $v$~and~$w$ occur in a common
    hyperedge, and
  \item equality of vertex variables, edge variables, and set
    variables.
  \end{enumerate}
  We use upper-case letters for set variables and lower-case letters
  for vertex and hyperedge variables.
\end{definition}

\newcommand{\mst}{MS$_2$}
\newcommand{\mso}{MS$_1$}

\begin{corollary}\label{nomso}
  Let~$\mcF{}\subseteq \hsmall{}$ be a hypergraph problem such that
  $\sim_{\mcF{}}$ has infinite index over~$\hsmall$.
  Then, there is no MS$_2$-formula that a hypergraph~$H\in\hsmall$
  satisfies if and only if~$H\in\mcF{}$.
\end{corollary}
\begin{proof}
  \looseness=-1 Let $F\subseteq\usmall{2}$ be \full{} such
  that~$\hyperg{F}=\mcF{}$ and assume, towards a contradiction, that
  there is an \mst{}-formula~$\varphi$ for hypergraphs such
  that~$\mcF{}=\{H\in\hsmall\mid H\text{ satisfies }\varphi\}$. We
  will turn~$\varphi$ into an \mso{}-formula~$\varphi^*$ for colored
  graphs such that a hypergraph~$H\in\hsmall$ satisfies~$\varphi$ if
  and only if all graphs $G\in\usmall{2}$ with~$\hyperg{G}=H$
  satisfy~$\varphi^*$. That is, $F=\{G\in\usmall{2}\mid G\text{
    satisfies }\varphi^*\}$. Courcelle's
  theorem~\citep[Theorem~6.3(2)]{CE12} shows that~$\varphi^*$ can be
  turned into a tree automaton~$\mathcal A_{\varphi^*}$ such that the
  parse tree of a graph~$G\in\usmall{2}$ is accepted by~$\mathcal
  A_{\varphi^*}$ if and only if
  $G$~satisfies~$\varphi^*$. Consequently, $A_{\varphi^*}$
  recognizes~$F$.  This, by \autoref{afhlemma},
  contradicts~$\sim_\mcF{}$ having infinite index.

  It remains to describe the transformation from~$\varphi$
  to~$\varphi^*$. To this end, recall that, by \autoref{liftbip} of
  \wcg{t}s, the color-1 vertices in a graph~$G\in\usmall{2}$
  with~$\hyperg{G}=H$ represent vertices of~$H$ while the color-2
  vertices in~$G$ represent hyperedges of~$H$. Hence, the vertex and
  hyperedge variables in~$\varphi$ both become vertex variables
  in~$\varphi^*$. Moreover, the formula~$\varphi^*$ makes sure that
  the graph~$G\in\usmall{2}$ satisfying~$\varphi^*$ is a \wcg{t}, that
  is, vertices of the same color are nonadjacent in~$G$ and each
  vertex in~$G$ has a color.  Thus, we let
  \begin{align*}
    \varphi^* &:= \text{all-labeled}\wedge \text{bipartite} \wedge \varphi',\\
    \text{all-labeled} &:= \forall v[\lab_1(v)\vee\lab_2(v)],\\
    \text{bipartite} &:= \forall v\forall
    w[(\lab_1(v)\wedge\lab_2(w))\vee(\lab_2(v)\wedge\lab_1(w))\vee\neg\adj(v,w)],
  \end{align*}
  where we obtain~$\varphi'$ by replacing terms in~$\varphi$ referring
  to hypergraphs by equivalent terms referring to incidence graphs.
  The term replacement translates incidence and adjacency of
  hypergraph objects into adjacency of the corresponding incidence
  graph vertices. Specifically, we replace the following hypergraph
  \mst{}-expressions on the left-hand side by the equivalent graph
  \mso{}-expressions on the right-hand side:
  \begin{align*}
    \inc(e,v) &\equiv \lab_2(e)\wedge\lab_1(v)\wedge\adj(e,v),\text{ and}\\\displaybreak[3]
    \adj(v,w) &\equiv \lab_1(v)\wedge\lab_1(w)\wedge\exists e[\lab_2(e)\wedge\adj(e,v)\wedge\adj(e,w)].
  \end{align*}
  Quantification over the vertices, hyperedges of a
  hypergraph~$H=(V,E)$ are realized by the term replacements
  \begin{align*}
    \exists v\in V[\psi]&\equiv \exists v[\lab_1(v)\wedge\psi],&
    \exists S\subseteq V[\psi]&\equiv \exists S[\forall x[x\notin S\vee \lab_1(x)]\wedge\psi],\\
    \exists e\in E[\psi]&\equiv \exists e[\lab_2(e)\wedge\psi],&    \exists S\subseteq E[\psi]&\equiv \exists S[\forall x[x\notin S\vee \lab_2(x)]\wedge\psi],\\
    \forall v\in V[\psi]&\equiv \forall v[\neg\lab_1(v)\vee\psi],&
    \forall S\subseteq V[\psi]&\equiv \forall S[\exists x[x\in S\wedge\neg\lab_1(x)]\vee\psi],\\
    \forall e\in E[\psi]&\equiv \forall e[\neg\lab_2(v)\vee\psi],&    \forall S\subseteq E[\psi]&\equiv \forall S[\exists x[x\in S\wedge\neg\lab_2(x)]\vee\psi].
  \end{align*}
\end{proof}

\section{Hypergraph Cutwidth is fixed-parameter linear}\label{sec5}
In this section, we use the Myhill-Nerode theorem for hypergraphs to
show that \hgcw{} is fixed-parameter linear. We first formally define
the problem.

Let $H=(V,E)$ be a hypergraph.  A \emph{linear layout} of $H$ is an
injective map~$l\colon V\to\mathbb R$ of vertices onto the real line.
The \emph{cut at position~$i\in\mathbb R$} in~$H$ with respect to~$l$,
denoted $\normtheta{l}{H}{i}$, is the set of hyperedges that contain
at least two vertices~$v,w$ such that~$l(v)<i<l(w)$.  We will also say
that $v$~is to the \emph{left} of~$i$ and that $w$~is to the
\emph{right} of~$i$. The \emph{cutwidth} of the layout $l$ is
\[\max_{i\in\mathbb R} {|\normtheta{l}{H}{i}|}.\]  The \emph{cutwidth} of
the hypergraph $H$ is the minimum cutwidth over all the linear layouts
of $H$. The hypergraph shown in \autoref{cwbeispiel} has cutwidth at
most three.  The \hgcw{} problem is defined as follows.

\decprob{\hgcw{}}{A hypergraph~$H=(V,E)$ and a natural
  number~$k$.}{Does $H$ have cutwidth at most~$k$?}{$k$}

\begin{figure}
  \centering
  \begin{tikzpicture}[shorten >=1pt,
    shorten <=1pt]
    \tikzstyle{edge} = [color=black,opacity=.2,line cap=round, line
    join=round, line width=19pt]

    \foreach \l in {1,...,12} {
      \node [blue] (v\l) at (\l,0) {}; }

    \begin{pgfonlayer}{background}
      \draw [edge] (v1.center) -- (v4.center); \draw [edge]
      (v5.center) -- (v8.center); \draw [edge] (v9.center) --
      (v12.center);
      \draw [edge, line width=25pt] (3,1) -- (v3.center) --
      (v2.center) -- (3,1) -- (11,1) -- (v11.center) -- (v12.center)
      -- (11,1);

      \draw [edge, line width=22pt] (v7.center) -- (7,1) -- (9,1) --
      (v9.center) -- (v10.center) -- (9,1);
    \end{pgfonlayer}
    \draw (9.5,-0.75) -- (9.5,1.75);
  \end{tikzpicture}
  \caption{The shown hypergraph has cutwidth at most three since the
    black line cuts a maximum number of hyperedges in the presented
    linear layout. Actually, it is possible to change the linear
    layout to see that the hypergraph has cutwidth two.}
  \label{cwbeispiel}
\end{figure}
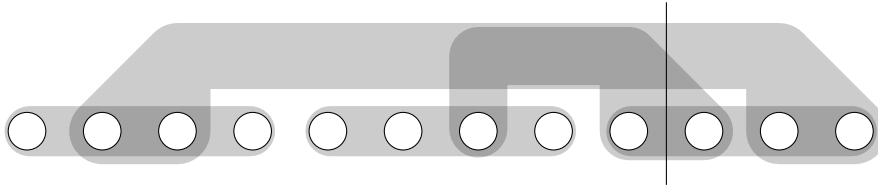
\noindent To solve \hgcw{} using the Myhill-Nerode theorem for
hypergraphs, in the remainder of this section we consider a
constant~$k$ and the class \brc{k} of all hypergraphs with cutwidth at
most $k$. We will solve \brc{k}
in linear time using \autoref{hisoclosed}. This will immediately yield
the main result of this section:

\begin{theorem}\label{easythm}
  \hgcw{} is fixed-parameter linear. Specifically, there is an
  algorithm that, when given a hypergraph~$H$ as hyperedge list and a
  constant~$k$, decides in linear time whether $H$~has cutwidth at
  most~$k$.
\end{theorem}

\noindent In order to use \autoref{hisoclosed} to prove
\autoref{easythm}, we first show that the hypergraphs in \brc{k} have
a constant upper bound on their \iw{}. Then, we show that the
\crc{}~$\congr$ has finite index. By \autoref{hisoclosed}, it then
follows that \brc{k} is solvable in~$f(k)\cdot n$ time, completing the
proof of \autoref{easythm}.

\begin{lemma}\label{twbound}
  Let $H$~be a hypergraph.  If $H$~has cutwidth at most~$k$, then
  \begin{enumerate}[i)]
  \item $H$~has \iw{} at most $\max \{k,1\}$, and
  \item the incidence graph of~$H$ has pathwidth at
  most~$k+1$.
  \end{enumerate}
\end{lemma}
\begin{proof}
  Suppose that $H=(V,E)$~has cutwidth at most~$k$.  Let $H'=(V,E')$
  denote the hypergraph obtained from~$H$ by removing all hyperedges
  of size at most one.  Consider a linear layout~$l$ of cutwidth at
  most~$k$ of the vertices of~$H'$.  Without loss of generality,
  assume that $l$~maps to the natural numbers~$[n]$ and
  let~$V=\{v_1,\dots,v_n\}$ be such that $l(v_i)=i$. We construct a
  path decomposition for the incidence graph~$\mathcal I(H')$ of~$H'$
  with the bags~$L_1, R_1 ,L_2, R_2, \dots , L_{n}, R_{n}$ that are
  connected by a path in this order.  For every~$i\in[n]$, let $L_i :=
  \normtheta{l}{H'}{i-1/2} \cup \{v_i\}$ and $R_i :=
  \normtheta{l}{H'}{i+1/2} \cup \{v_{i}\}$, that is, $L_i$~contains
  $v_i$~and all hyperedges cut at~$i-1/2$, while $R_i$~contains
  $v_i$~and all hyperedges cut at~$i+1/2$. Herein, recall that the
  hyperedges of~$H'$ are vertices in~$\mathcal I(H')$.  We now prove
  that this is a path decomposition for~$\mathcal I(H')$. 

  First, we show that each edge of~$\mathcal I(H')$ is contained in at
  least one bag.  Let $\{v_i,e\}$ be any edge in~$\mathcal I(H')$ for
  some vertex~$v_i\in V$ and a hyperedge~$e\in E'$. We show that~$v_i$
  and~$e$ occur together in at least one bag. Since $v_i\in e$ and
  $|e|\ge 2$, the hyperedge~$e$ contains at least one vertex to the
  left or to the right of~$v_i$. Hence, we have $e\in
  \normtheta{l}{H'}{i-1/2}$ or $e\in \normtheta{l}{H'}{i+1/2}$.
  Therefore, it holds that $e\in R_{i}$ or $e\in L_i$. Since $v_i\in
  R_{i} \cap L_i$, the vertices~$v_i$ and~$e$ occur together in at
  least one bag.
 
  Now, we show that the bags containing a vertex of~$\mathcal I(H')$
  induce a subpath in this path decomposition.  Obviously, each
  vertex~$v_i\in V$ is contained in two bags of the path
  decomposition: in~$L_i$ and~$R_{i}$. These bags are consecutive and
  thus induce a path.  Finally, consider a hyperedge~$e\in E'$. It
  occurs in all bags~$R_i, L_{i+1}, R_{i+1}, \dots, L_{j-1}, R_{j-1},
  L_j$, where $v_i$~is the leftmost vertex in the layout~$l$ occurring
  in~$e$ and $v_j$~is the rightmost vertex in~$l$ occurring in~$e$.
  These bags are all consecutive on the path and, thus, induce a path.
   The width of this path decomposition is $\max_{0\le i\le n}
  |\normtheta{l}{H'}{i+1/2}| \le k$. 

  (i) To obtain a tree decomposition
  for~$\mathcal I(H)$ from the path decomposition of~$\mathcal I(H')$,
  we only need to take care of hyperedges of size at most one.  For
  every hyperedge~$e\in E$ of size one, add a new bag $\{e, v\}$,
  where $v$~is the unique vertex contained in~$e$, and make it
  adjacent to an arbitrary bag containing~$v$.  For every empty
  hyperedge~$e\in E$, add a new bag~$\{e\}$, and make it adjacent to
  an arbitrary bag.  In this way, we obtain a tree decomposition for
  the \br{}~$\mathcal I(H)$ of~$H$ of width at most~$\max \{k,1\}$.
  Thus, $H$ has \iw{} at most $\max \{k,1\}$.

  (ii) To obtain a path decomposition for~$\mathcal I(H)$ from the path
  decomposition of~$\mathcal I(H')$, one can proceed similarly: For every
  hyperedge~$e\in E$ of size at most one, choose some bag~$B$ of size
  at most~$k+1$ with $e\subseteq B$ and add the bag~$B\cup\{e\}$
  as its neighbor to the path decomposition. Such a bag~$B$ exists since the
  width of the path decomposition of~$H'$ is~$k$.  The resulting path
  decomposition will contain bags of size~$k+2$ and, thus, has width
  $k+1$.  
\end{proof}

\noindent To obtain a linear-time algorithm for \brc{k} using
\autoref{hisoclosed} and thus proving \autoref{easythm}, it remains to
prove that the \crc{}~$\congr$ of \brc k has finite index
over~$\hlarge$ for all~$t\leq k+1$.

\bigskip\noindent To show that $\congr$~has finite index
over~$\hlarge$, we show that, given a $t$-boundaried hypergraph~$G$,
only a finite number of bits of information about a $t$-boundaried
hypergraph~$H$ is needed in order to decide whether $G\oplush
H\in\brc{k}$. To this end, we employ the method of test
sets~\citep[Section~12.7]{DF13}: let $\mathcal T$ be a set of objects
called \emph{tests} (we will formally define a test later). A \bndg t
can \emph{pass} a test. For \bndh ts $G_1$ and~$G_2$, let $G_1\simtest
G_2$ if and only if $G_1$ and~$G_2$ pass the same subset of tests
in~$\mathcal T$. Obviously, $\simtest$ is an equivalence relation. Our
aim is to find a set $\mathcal T$~of tests such that $\simtest$
refines~$\congr$ (that is, $G_1\simtest G_2$ implies $G_1\congr
G_2$). Then, if $\simtest$~has finite index, so does $\congr$. To
show that $\simtest$~has finite index, we show that we can find a
\emph{finite} set $\mathcal T$ such that $\simtest$ refines~$\congr$.

Intuitively, we will define, for a hypergraph~$H$, an $H$-test that a
hypergraph~$G$ satisfies if~$G\oplush H\in\brc{k}$. We define the test
so that it contains only the necessary information of~$H$ and so that
we can later shrink all tests to equivalent tests of constant size. We
now formally define a test for~\brc{k}. The definition is illustrated
in \autoref{fig:glue2} and, after the definition, we give an intuitive
description.

\begin{definition}\label{testdef}
  A \emph{size-$n$ test}~$T$ for \brc{k} over~$\hlarge$ is a
  triple~$(\pi,S,k)$, where
  \begin{itemize}
  \item $\pi\colon\{1,\dots,t\}\to\{1,\dots,n\}$~is a map of boundary
    labels to integer positions, and
  \item $S=(S_0,S_1,\dots,S_n)$~is a sequence of
    pairs~$S_i=(w_i,E_i)\in \mathbb \{0,\dots,k\}\times
    2^{\{1,\dots,t\}}$ such that if~$\ell\in E_p$ and~$\ell\in E_q$,
    then $\ell\in E_i$ for all~$i\in\{p,\dots,q\}$.
  \end{itemize}
  Now, let $G$~and~$H$ be \bndh{t}s such that $G\oplush H\in\brc k$
  and $l:V\to\mathbb R$~be a linear layout for~$G\oplush H$ with
  minimum cutwidth, which, without loss of generality, maps vertices
  of the $n$-vertex hypergraph~$H$ to the integer
  positions~$\{1,\dots,n\}$ and the non-boundary vertices of~$G$ to
  non-integer positions.

  We define an \emph{$H$-test}~$T=(\pi,S,k)$ for \brc{k} of size~$n$
  as follows: for a vertex~$v\in\partial(H)$ with label~$\ell$, set
  $\pi(\ell):=l(v)$. Finally, for $i\in\{0,\dots,n\}$, we
  define~$S_i:=(w_i, E_i)$ with
  \begin{itemize}
  \item $w_i$ being the number of unlabeled hyperedges in~$H$ containing
    vertices~$v,w$ of~$H$ with $l(v)\leq i< l(w)$, and
  \item $E_i$ being the set of labels of hyperedges in~$H$ containing
    vertices~$v,w$ of~$H$ with $l(v)\leq i\leq l(w)$.
  \end{itemize}
  
\end{definition}

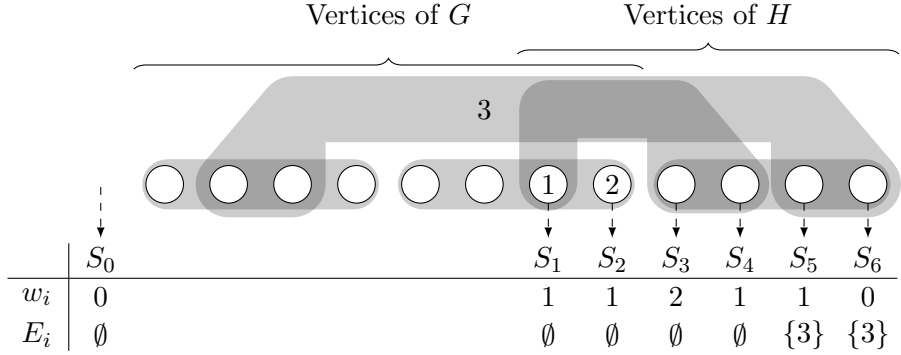
\begin{figure}[t]\centering
  \begin{tikzpicture}[x=0.85cm,shorten >=1pt, shorten <=1pt, on grid, auto]
    \tikzstyle{edge} = [color=black,opacity=.2,line cap=round, line join=round, line width=19pt]    

    \foreach \l in {1,...,12} {
      \node [blue, inner sep=1pt] (v\l) at (\l,0) {\ifnum\l=7{1}
        \else\ifnum\l=8{2}
        \fi
        \fi};
    }

    \node (l) at (6,1) {3};
    \begin{pgfonlayer}{background}
      \draw [edge] (v1.center) -- (v4.center);
      \draw [edge] (v5.center) -- (v8.center);
      \draw [edge] (v9.center) -- (v12.center);
      \draw [edge, line width=25pt] (3,1) -- (v3.center) -- (v2.center) -- (3,1) -- (11,1) -- (v11.center) -- (v12.center) -- (11,1);

      \draw [edge, line width=22pt] (v7.center) -- (7,1) -- (9,1) -- (v9.center) -- (v10.center) -- (9,1);

      \draw [decorate, decoration={name=brace,amplitude=5}]  (0.5,1.5) -- (8.5,1.5) node [midway, yshift=0.5cm] {Vertices of $G$};
      \draw [decorate, decoration={name=brace,amplitude=5}]  (6.5,1.7) -- (12.5,1.7)  node [midway, yshift=0.3cm] {Vertices of $H$};

      \node (S0) at (0,-1) {$S_{0}$};
      \node (w0) at (0,-1.5) {$0$};
      \node (E0) at (0,-2) {$\emptyset$};
      \draw [-latex,dashed] (0,0) -- (0,-0.75);

      \foreach \i in {1,...,6} {
        \pgfmathparse{\i+6}\let\j=\pgfmathresult;
        \node (S\i) at (\j,-1) {$S_{\i}$};
        \node (w\i) at (\j,-1.5) {\ifnum\i=0{0}
          \else\ifnum\i<3{1}
          \else\ifnum\i<4{2}
          \else\ifnum\i<5{1}
          \else\ifnum\i<6{1}
          \else{0}
          \fi
          \fi
          \fi
          \fi
          \fi};
        \node (E\i) at (\j,-2) {\ifnum\i=5{$\{3\}$}
          \else\ifnum\i=6{$\{3\}$}
          \else$\emptyset$
          \fi
          \fi};
        \draw [-latex,dashed] (\j,0) -- (\j,-0.75);
      }

      \node (wi) at (-1,-1.5) {$w_i$};
      \node (Si) at (-1,-2) {$E_i$};
      
      \draw (-1.5,-1.25) -- (12.5,-1.25);
      \draw (-0.5,-2.25) -- (-0.5,-0.75);
    \end{pgfonlayer}
  \end{tikzpicture}
  \caption{Construction of the $H$-test illustrated using the glued
    hypergraph $G\oplush H$, where $G$~and~$H$ are the hypergraphs
    shown in \autoref{hggluefig}(a) and \autoref{hggluefig}(b),
    respectively. That is, $G$~and~$H$ have only the vertices
    labeled~1 and~2 in common and both have a hyperedge with
    label~$3$. The vertices of~$H$ are to be understood as lying at
    the positions~$\{1,2,\dots,6\}$ and the non-boundary vertices
    of~$G$ lie in the open interval~$(0,1)$.}
  \label{fig:glue2}
  \label{fig:hg-glue}
\end{figure}

\noindent The goal of \autoref{testdef} is that if a hypergraph~$G$
passes an $H$-test for \brc{k}, then $G\oplush H\in\brc k$. More
precisely, we want that if a hypergraph~$G$ passes an $H$-test, then
$G\oplush H$~has a linear layout~$l$ of cutwidth at most~$k$ that lays
out the vertices of~$H$ in the same way as the layout used to create
the $H$-test. Of course, the $H$-test does not record the precise
structure of~$H$ but only the most important information:

Assume that we want to verify that the cutwidth of the layout~$l$
of~$G\oplush H$ is at most~$k$ without knowing~$H$ but only
knowing~$G$ and the~$H$-test. Then, for any non-integer position~$i$,
the value~$w_{\lfloor i\rfloor}$ counts the unlabeled hyperedges
of~$H$ cut at~$i$. Thus, to the size of any cut for~$G$ at
position~$i\in\mathbb R\setminus\mathbb N$, we have to add the
value~$w_{\lfloor i\rfloor}$. For labeled hyperedges of~$H$, things
are more difficult: they contain vertices of~$G\oplush H$ that
originate from~$G$ as well as from~$H$. Since an $H$-test corresponds
to a fixed layout for~$H$, to count a hyperedge with label~$\ell$
of~$G\oplush H$ that is cut at some position, it is sufficient to know
the vertices of the hyperedge with label~$\ell$ in~$G$ and the
positions of the leftmost and the rightmost vertex of~$H$ contained in
the hyperedge with label~$\ell$ in~$H$. However, in order to easier
shrink all tests to constant size later, we chose a more convenient
way to keep this information in the $H$-test: for any position~$i$
between the leftmost and the rightmost vertex of a hyperedge~$e$
in~$H$ with label~$\ell$, we have $\ell\in E_i$.  We now precisely
define what it means to pass a test.

\begin{definition}\label{passtest}
  Let $G=(V,E)$~be a \bndh t and $T=(\pi,S,k)$~be a test of size~$n$,
  where $S=(S_0,\dots,S_n)$ and $S_i=(w_i,E_i)$.

  A \emph{$T$-compatible layout for~$G$} is an injective
  function~$f\colon V\to\mathbb R$ such that each
  vertex~$v\in \partial(G)$ with label~$\ell$ is mapped to~$\pi(\ell)$
  and such that each vertex~$v\in V\setminus\partial(G)$ is mapped
  into some open interval~$(i,i+1)$ for $0\leq i\leq n$.

  For a hyperedge~$e$ in~$G$, we define the \emph{positions of~$e$} as
  \[
  \Pos(e):=
  \begin{cases}
    \{f(v)\mid v\in e\}&\text{if $e$ is unlabeled},\\
    \{f(v)\mid v\in e\}\cup\{i\mid \ell\in E_i\}&\text{if $e$ has label~$\ell$}.
  \end{cases}
  \]
  The \emph{\joint{} cut at~$i$} in~$G$ with respect to~$f$ is the
  set~$\mytheta{f}{G}{i}$ of hyperedges~$e$ of~$G$ for which there
  are positions~$j,k\in\Pos(e)$ with $j<i<k$.  The \emph{\joint{}
    cutwidth of~$f$} is
\[\max_{i\in \mathbb R\setminus\mathbb N}
  (|\mytheta{f}{G}{i}|+w_{\lfloor i\rfloor}).\]
  Finally, \emph{$G$~passes the test~$T$} if there is a $T$-compatible
  layout~$f$ for~$G$ whose \joint{} cutwidth is at most~$k$.
\end{definition}

\noindent We can now show that, indeed, if two graphs satisfy the same
tests, then they are equivalent under~$\congr$. We will then show
that, actually, there is only a finite set of pairwise nonequivalent
tests, thus showing that~$\congr$ has finite index.

\begin{lemma}\label{firstrefinement}
  For $\mathcal T$ being the set of all tests for~\brc{k}, the equivalence
  relation $\simtest$ refines~$\congr$.
\end{lemma}

\noindent To prove \autoref{firstrefinement}, we show that if two \bndh{t}s
$G_1,G_2$ pass the same subset of tests of~$\mathcal T$, then, for all
\bndh{t}s~$H$, $G_1\oplush H\in\brc k$ if and only if $G_2\oplush
H\in\brc k$.  The proof is based on the following two claims.

\begin{cla}\label{cla1app}
  If $G_1\oplush H\in\brc k$, then $G_1$ passes
    some $H$-test.
\end{cla}

\begin{cla}
  \label{cla2app} If $G_2$ passes any $H$-test, then $G_2\oplush
    H\in\brc k$.
\end{cla}

\noindent From these two claims, \autoref{firstrefinement} then easily
follows: let $H$~be a \bndh{t} such that $G_1\oplush H\in\brc k$. By
\autoref{cla1app}, $G_1$~passes some $H$-test~$T$. Since $G_1$
and~$G_2$ pass the same tests, also $G_2$ passes~$T$. By
\autoref{cla2app}, it follows that $G_2\oplush H\in \brc k$. The
reverse direction is proved symmetrically. It only remains to prove
\autoref{cla1app} and~\autoref{cla2app}.

\begin{proof}[Proof of \autoref{cla1app}]
  Let $T$~be the $H$-test obtained from an optimal layout~$l$
  of~$G_1\oplush H$, which, without loss of generality, maps the
  vertices of the $n$-vertex graph~$H$ to the integer
  positions~$\{1,\dots,n\}$ and the vertices
  of~$V(G_1)\setminus \partial(G_1)$ to non-integer positions in the
  interval~$(0,n+1)$. Then, $l$ obviously~is a $T$-compatible layout
  for~$G_1$. We show that the \joint{} cutwidth $\max_{i\in \mathbb
    R\setminus\mathbb N} (|\mytheta{l}{G_1}{i}|+w_{\lfloor i\rfloor})$
  of~$l$ from \autoref{passtest} is at most~$k$. 

  To this end, for an~$i\in\mathbb R\setminus\mathbb N$, consider the
  set~$A:=\normtheta{l}{G_1\oplush H}{i}$ of hyperedges of~$G_1\oplush
  H$ containing two vertices~$v,w$ with $l(v)<i<l(w)$. Since
  $G_1\oplush H\in\brc{k}$, we have $|A|\leq k$. Thus, it is
  sufficient to show that $|\mytheta{l}{G_1}{i}|+w_{\lfloor i\rfloor}
  \leq |A|$. We partition~$A$ into two sets~$B$ and~$C$, where $B$~are
  the unlabeled hyperedges in~$H$.

  Since $i\notin\mathbb N$, by \autoref{testdef}, $w_{\lfloor i
    \rfloor}$  counts exactly the hyperedges in~$B$.
  It remains to show that $|\mytheta{l}{G_1}{i}|\leq |C|$. Recall from
  \autoref{passtest} that $\mytheta{l}{G_1}{i}$~is the set of
  hyperedges~$e$ of~$G_1$ for which $\Pos(e)$~contains two
  positions~$j,k$ with~$j<i<k$. If $e$~is unlabeled, then, by
  \autoref{passtest} of~$\Pos(e)$, the hypergraph~$G_1$ contains
  vertices~$v,w$ with~$l(v)=j$ and~$l(w)=k$. Since~$e,v$, and~$w$ are
  also in~$G_1\oplus H$, we have~$e\in C$.

  If $e$~is labeled, then $G_1\oplush H$ instead of~$e$ contains a
  hyperedge~$e'\supseteq e$.  Now, since $j\in\Pos(e)$, the
  hypergraph~$G_1$ contains a vertex~$v$ with~$l(v)=j$ or $\ell\in
  E_j$, which, by \autoref{testdef}, implies that there is a hyperedge
  with label~$\ell$ containing a vertex~$v$ in~$H$ with~$l(v)\leq
  i$. Likewise, $G_1$ contains a vertex~$w$ with~$l(w)=k$ or $\ell\in
  E_k$, which implies that there is a hyperedge with label~$\ell$
  containing a vertex~$w$ in~$H$ with~$k\leq l(w)$. In all cases, we
  have that~$l(v)<i<l(w)$. Since the hyperedge~$e'$ of~$G_1\oplush H$
  contains~$v$ and~$w$, we get $e'\in C$.  
\end{proof}

\begin{proof}[Proof of \autoref{cla2app}]
  Let $T$~be an $H$-test~$T$ obtained from a linear layout~$l$ of
  cutwidth~$k$ for some $G^*\oplush H$ and assume that~$G_2$
  passes~$T$. Then, there is a $T$-compatible layout~$f$ for~$G_2$
  with \joint{} cutwidth at most~$k$. First note that $l$~and~$f$
  agree on the layout of vertices in~$\partial(G_2)$ and~$\partial(H)$
  and that, apart from these, $f$~lays out vertices at non-integral
  positions, whereas $l$ lays out vertices of~$H$ at integral
  positions. Because of this, in a layout~$g$ for~$G_2\oplush H$ that
  lays out vertices~$v$ of~$H$ at position~$l(v)$ and vertices~$v$
  of~$G_2$ at position~$f(v)$, every two vertices in~$G_2\oplush H$
  are laid out at distinct positions by~$g$. Hence, $g$ is injective
  and, therefore, a layout.

  We show that $g$~is a layout of cutwidth at most~$k$ for~$G_2\oplush
  H$. That is, we show $\max_{i\in\mathbb R} {|\normtheta{g}{G_2\oplush
      H}{i}|}\leq k$. To this end, note that
  \[
  \max_{i\in\mathbb R} {|\normtheta{g}{G_2\oplush H}{i}|}\leq
  \max_{i\in\mathbb R\setminus\mathbb N} {|\normtheta{g}{G_2\oplush H}{i}|},
  \]
  since for every~$i\in\mathbb N$, we have $\normtheta{g}{G_2\oplush
    H}{i}\subseteq \normtheta{g}{G_2\oplush H}{i+\varepsilon}$
  for~$0<\varepsilon<1$ chosen so that no vertex is mapped by~$g$ to
  the interval~$(i,i+\varepsilon]$.  That is, we only have to show
  that, for each~$i\in\mathbb R\setminus\mathbb N$, we have
  $|\normtheta{g}{G_2\oplus H}{i}|\leq|\mytheta{f}{G_2}{i}|+w_{\lfloor
    i\rfloor}$, since $f$~is a layout for~$G_2$ with \joint{}
  cutwidth~$\max_{i\in\mathbb R\setminus\mathbb
    N}(|\mytheta{f}{G_2}{i}|+w_{\lfloor i\rfloor})\leq k$.

  For some position~$i\in\mathbb R\setminus\mathbb N$, consider the
  set~$A:=\normtheta{g}{G_2\oplush H}{i}$ of hyperedges of~$G_2\oplush
  H$ containing vertices~$v,w$ with $g(v)<i<g(w)$ and let it be
  partitioned into two sets~$B$ and~$C$, where $B$~contains the
  unlabeled hyperedges of~$H$.  We show that~$|A|\leq w_{\lfloor
    i\rfloor}+|\mytheta{f}{G_2}{i}|$. By \autoref{testdef}, we clearly
  have $|B|\leq w_{\lfloor i\rfloor}$. Hence, it remains to show
  that~$|C|\leq |\mytheta{f}{G_2}{i}|$.

  To this end, let $e\in C$ be a hyperedge. If $e$ contains only
  vertices of~$G_2$, then for any vertex~$v\in e$, we
  have~$g(v)=f(v)\in\Pos (e)$. Furthermore, if $e$~contains a
  vertex~$v$ of~$H$, then $e$~has a label~$\ell$ and~$H$ has a
  hyperedge with label~$\ell$ containing~$v$. Hence, in this case, we
  have~$\ell\in E_{l(v)}$. It follows that $g(v)=l(v)\in\Pos(e')$ for
  the hyperedge~$e'\subseteq e$ of~$G_2$ with label~$\ell$. Hence, for
  any vertex~$v\in e$, we have~$g(v)\in\Pos(e')$. Since $e$~contains
  vertices~$v,w$ with~$g(v)<i<g(w)$, it follows that~$\Pos(e')$
  contains positions~$j=g(v)$ and~$k=g(w)$ with $j<i<k$ and,
  therefore, $e'\in \mytheta{f}{G_2}{i}$.  
\end{proof}

\noindent Towards our goal of showing that~$\congr$ has finite index,
\autoref{firstrefinement} shows a set of tests~$\mathcal T$ such that
$\simtest$~refines $\congr$, where two hypergraphs are equivalent with
respect to~$\simtest$ if and only if they pass the same subset of
tests of~$\mathcal T$.  However, since the set~$\mathcal T$ is
infinite, we cannot yet conclude that~$\simtest$ and, therefore,
$\congr{}$ has finite index. The following lemma will, for every
test~$T\in\mathcal T$, find a test~$T'\in\mathcal T$ such that a
hypergraph~$G$ passes~$T'$ if and only if it passes~$T$ and such
that~$T'$ has size at most~$(2t+1)(t+1)(2k+2)$. Thus, the equivalence
relation~$\sim_{\mathcal T'}$ for~$\mathcal T'$ being the set of all
tests of size~$(2t+1)(t+1)(2k+2)$ is the same as~$\simtest$ and,
consequently, refines~$\congr$. Since there is only a constant number
of tests of size~$(2t+1)(t+1)(2k+2)$ for constant~$k$ and $t\leq k+1$, the
size of~$\mathcal T'$ is constant. Since $\sim_{\mathcal T'}$ and,
therefore, $\congr$~has at most $2^{|\mathcal T'|}$~equivalence
classes, it follows that $\congr$~has finite index. Thus, the
following lemma finishes our proof of \autoref{easythm}.

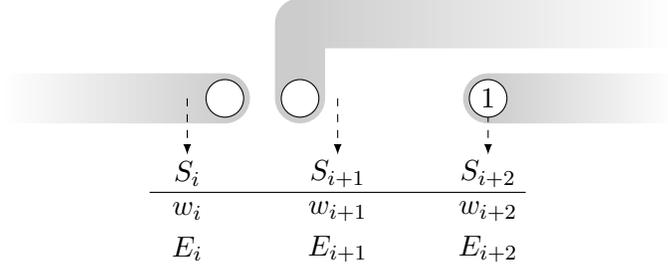
\begin{figure}
  \centering
  \begin{tikzpicture}
    \tikzstyle{edge} = [color=black,opacity=.2,line cap=round, line join=round, line width=19pt]    

    \node (S1) at (-0.5,-1) {$S_{i}$};
    \node (w1) at (-0.5,-1.5) {$w_i$};
    \node (E1) at (-0.5,-2) {$E_i$};
    \draw [-latex,dashed] (-0.5,0) -- (-0.5,-0.75);

    \node (S1) at (1.5,-1) {$S_{i+1}$};
    \node (w1) at (1.5,-1.5) {$w_{i+1}$};
    \node (E1) at (1.5,-2) {$E_{i+1}$};
    \draw [-latex,dashed] (1.5,0) -- (1.5,-0.75);

    \node (S1) at (3.5,-1) {$S_{i+2}$};
    \node (w1) at (3.5,-1.5) {$w_{i+2}$};
    \node (E1) at (3.5,-2) {$E_{i+2}$};
    \draw [-latex,dashed] (3.5,0) -- (3.5,-0.75);

    \node[blue] (v1) at (0,0) {};
    \node[blue] (v2) at (1,0) {};
    \node[blue] (v3) at (3.5,0) {$1$};

    \draw (-1,-1.25) -- (4,-1.25);

    \begin{pgfonlayer}{background}
      \draw[edge,path fading=west] (v1.center)--(-3,0) ++(0,-20pt);
      \draw[edge,path fading=east] (v2.center)--(1,1)--(6,1);
      \draw[edge,path fading=east] (v3.center)--(6,0) ++(0,-20pt);
    \end{pgfonlayer}
  \end{tikzpicture}
  \caption{Shown are two unlabeled vertices and one labeled vertex of
    a graph~$G$ laid out according to a $T$-compatible layout for
    some test~$T=(\pi,S,k)$. That is, the label~1 is mapped to the
    integer position~$i+2$ by~$\pi$, while the others vertices are
    laid out at non-integer positions. Assume that~$S_i=S_{i+1}$ and
    that no label is mapped to position~$i+1$ by~$\pi$. Then, we can
    assume that no vertex of~$G$ lies in~$[i+1,i+2)$: moving it to
    $(i,i+1)$~would yield a $T_1$-compatible layout with equal
    \joint{} cutwidth. The \joint{} cutwidth will also not be altered
    by deleting~$S_{i+1}$ or adding copies of~$S_{i}$
    behind~$S_{i}$.  }
  \label{fig:testred}
\end{figure}

\begin{lemma}\label{shorttests}
  Let $G$ be a \bndh{t}. For every test~$T_1$, there is a test~$T_2$
  of size~$(2t+1)(t+1)(2k+2)$ such that $G$~passes~$T_1$ if and only if
  $G$~passes~$T_2$.
\end{lemma}

\begin{proof}
  \looseness=-1 Let the size of the test~$T_1=(\pi,S,k)$ be~$n$. For
  $E\subseteq\{1,\dots,t\}$, we call a maximal subsequence
  $S_j=(E_j,w_j),\dots,S_k=(E_k,w_k)$ of~$S$ with $E=E_j=\dots=E_k$ a
  \emph{strait}. We first show that there are at most $2t+1$ straits,
  and then show that we can shorten each strait to length at
  most~$(t+1)(2k+2)$ by removing some elements from~$S$ without
  changing the satisfiability of the test.

  For a label~$\ell\subseteq\{1,\dots,t\}$, let $I_\ell:=\{i\leq n
  \mid \ell\in E_i\}$. By \autoref{testdef}, each $I_\ell$ for some
  label~$\ell\in\{1,\dots,t\}$ is an interval of the natural numbers
  with a minimum element and a maximum element, which we both call
  \emph{events}. Hence, the $I_\ell$ for all~$\ell\in\{1,\dots,t\}$ in
  total have at most $2t$~events. Since straits can only start at an
  event or at~$S_0$, and since only one strait can start at a fixed event, it
  follows that $S$~is partitioned into at most $2t+1$ straits.

  It remains to shorten the straits. To this end, we apply data
  reduction rules already used by \citet[Theorem~12.7.5]{DF13} for the
  cutwidth problem on graphs.  Let $S_j=(E,w_j),\dots,S_k=(E,w_k)$~be
  a strait in~$T_1$. We call a maximal subsequence of the~$w_i$ of the
  strait such that~$\pi$ maps no boundary label to~$i$ a \emph{load
    pattern}. Hence, each strait decomposes into at most~$t+1$ load
  patterns, each of which we will shorten to length at most~$2k+2$.

  To this end, first observe that if the test~$T_1$ passed by~$G$
  contains a pair~$S_i=(E,w_i)$, then $G$~also passes the test
  obtained from~$T_1$ by replacing~$S_i$ by~$S_i'=(E,w_i')$
  with~$w_i'\leq w_i$. Moreover, assume that, as illustrated in
  \autoref{fig:testred}, $T_1$~contains two pairs
  $S_i=(E,w_i),S_{i+1}=(E,w_{i+1})$ with $w_i=w_{i+1}$ such that~$\pi$
  maps no boundary label to~$i+1$. Then $G$~passes the test obtained
  from~$T_1$ by removing~$S_{i+1}$. Moreover, $G$~then also passes the
  test obtained from~$T_1$ by adding a copy of~$S_{i}$ behind~$S_{i}$.

  Based on these observations, \citet[Theorem~12.7.5]{DF13} give a proof
  that the following three data reduction rules applied to a load
  pattern~$s$ of the strait $S_j,\dots,S_k$ turn~$T_1$ into a
  test~$T_2$ that~$G$ passes if and only if it passes~$T_1$:
  \begin{enumerate}[(R1)]
  \item If $s=(\dots,w_i,w_{i+1},w_{i+2},\dots)$ such that~$w_i\leq
    w_{i+1}\leq w_{i+2}$ or $w_i\geq w_{i+1}\geq w_{i+2}$, then
    delete~$S_{i+1}$.
  \item If $s=(\dots,a,s_{i},\dots,s_{i'},b,\dots)$ such that each
    of~$s_{i},\dots,s_{i'}$ is at least $\max(a,b)$, then
    replace~$S_{i},\dots,S_{i'}$ by~$S^*:=(E,w)$, where~$w$ is the
    maximum of~$s_{i},\dots,s_{i'}$.
  \item If $s=(\dots,a,s_{i},\dots,s_{i'},b,\dots)$ such that each
    of~$s_{i},\dots,s_{i'}$ is at most $\min(a,b)$, then
    replace~$S_{i},\dots,S_{i'}$ by~$S^*:=(E,w)$, where~$w$ is the
    minimum of~$s_{i},\dots,s_{i'}$.
  \end{enumerate}
  \citet[Theorem~12.7.5]{DF13} show that a load pattern, to which none
  of the rules apply, has length at most~$2k+2$.  
\end{proof}

\paragraph{Historical Remarks.}  The above results about reduced load
patterns in the construction of test sets were first proved by
\citet{FellowsA89,AbrahamsonF91} in the context of proving that, for
simple graphs, the property $\mathcal{P}_{k}$ of having cutwidth
bounded by~$k$ has finite index over $\ulargep$ for all fixed~$k$
and~$t$.  An essentially equivalent notion, termed \textit{typical
  sequences}, was introduced independently by \citet{BodlaenderK96} in
the context of linear-time dynamic programming algorithms for
\textsc{Pathwidth} and \textsc{Treewidth}.  Such sequences are also
implicit in early work of \citet{LagergrenA91}.

\pagebreak
\section{Hypertree width and variants}\label{sec6}

\noindent In in this section, we show a negative application of our
hypergraph Myhill-Nerode analog to \ghtw~\citep{GottlobMS09}.  First,
we precisely define the problem.

Let $H$~be a hypergraph. Generalized hypertree width is defined with
respect to tree decompositions of the primal graph~$\mathcal{G}(H)$,
however, the width of the tree decompositions is measured differently.
Suppose $H$~has no isolated vertices (otherwise, remove them).  A
\emph{cover} of a bag is a set of hyperedges such that each vertex in
the bag is contained in at least one of these hyperedges.  The
\emph{cover width} of a bag is the minimum possible number of
hyperedges covering it.  The \emph{cover width} of a tree
decomposition is the maximum cover width of any bag in the
decomposition.  The \emph{generalized hypertree width} of~$H$ is the
minimum cover width over all tree decompositions of~$\mathcal G(H)$.

\decprob{\textsc{Generalized Hypertree Width}}
{A hypergraph~$G=(V,E)$ and a natural number~$k$.}
{Does $G$~have generalized hypertree width at most~$k$?}

\begin{comment}
  \begin{definition}
    Let $H$~be a hypergraph without isolated vertices and consider a
    tree decomposition of~$\mathcal{G}(H)$.
    %
    %
    %
    A \emph{cover} of a bag is a set of hyperedges such that each
    vertex in the bag is contained in at least one of these
    hyperedges.  The \emph{cover width} of a bag is the minimum number
    of hyperedges covering it.  The \emph{cover width} of a tree
    decomposition is the maximum cover width of any bag in the
    decomposition.  The \emph{generalized hypertree width} of~$H$ is
    the minimum cover width over all tree decompositions of~$\mathcal
    G(H)$.
  \end{definition}
\end{comment}
\noindent Since \ghtw{} is NP-hard
for~$k=3$~\citep{GottlobMS09}, %
it is natural to search for non-standard parameters with respect to
which the problem is fixed-parameter tractable~\citep{Nie10,KN12,FJR13}.  While
it is known that the generalized hypertree width of a hypergraph is at
most the \iw{} plus one~\cite{FominGT09}, the \iw{} may be arbitrarily
large even for hypergraphs with %
hypertree width one, since adding a universal hyperedge to any
hypergraph reduces its hypertree width to one.  Therefore, one could
hope for positive results with respect to \iw{}.  However, we will
show that \ghtw{} cannot be solved by finite tree automata on tree
decompositions of incidence graphs:

\begin{theorem}\label{THM:HTW}
  Let \kghtw{k} be the set of hypergraphs with generalized hypertree
  width at most~$k$.  The \crc{} $\sim_{\kghtw{k}}$ does not have
  finite index over~$\hsmall{}$ for~$k=4$ and~$t\geq 41$.
\end{theorem}

\noindent\looseness=-1 By \autoref{nomso}, it follows that \kghtw{k} is not
expressible in monadic second-order logic. Moreover, the construction
we use in the proof leads us to conjecture that, actually, the problem
might turn out to be W[1]-hard, as did \textsc{Bandwidth} and
\textsc{Triangulating Colored Graphs} after it was shown that they do
not have finite index~\citep{BFH94}.

We will discuss this after proving \autoref{THM:HTW}. Moreover, after
proving \autoref{THM:HTW}, we will discuss that the theorem also holds
for the problem variants \textsc{Hypertree Width} and
\textsc{Fractional Hypertree Width}.

\newcommand{\cB}{\mathcal{B}}

\bigskip\noindent To prove \autoref{THM:HTW}, we apply
\autoref{hisoclosed2}: for every $n\geq 1$, we give a construction of a
$t$-boundaried hypergraph~$H_n$ whose incidence graph allows for a
tree decomposition of width~$t-1$, of which one bag
contains~$\partial{H_n}$. Then we show that $H_n \oplush H_m$ has
generalized hypertree width~$4$ if and only if~$n=m$.  This implies
that the \crc{}~$\sim_{\khtw4}$ has infinite index over~$\hlarge$ and,
by \autoref{hisoclosed2}, it follows that it has infinite index also
over~$\hsmall{}$.

\begin{construction}
  For every $n\ge 1$, we construct a $t$-boundaried hypergraph~$H_n$
  with $t=28$, generalized hypertree width~$4$, and \iw{} at
  most~$12$.  The vertex set of~$H_n$ is $V :=A\cup B\cup C\cup D\cup
  S\cup T\cup X$, where $A :=\{a,y\}, B:=\{b,z\}, C:=\{c,y\},
  D:=\{d,z\}, S:=\{s_1, \dots, s_8\}, T :=\{t_1, \dots, t_8\}$ and $X
  :=\{x_1, \dots, x_{6n}\}$.  The hyperedge set of~$H_n$ is $E :=
  \{A,B,C,D\}\cup B_S\cup \{S_c,S_d,\allowbreak S_y,S_z\} \cup B_T
  \cup \{T_a,T_b,\allowbreak T_y,T_z\} \cup \{E_{3i}, E_{3i+1}: 1\le
  i< 2n\} \cup \{E_{i,i+1}: 1\le i < 6n\}$, where
  {\allowdisplaybreaks\begin{align*}
      B_S &\mathrlap{\text{ is the set of all possible binary hyperedges on $S$},}\\
      S_c &:= \{c,s_1,s_2\}& S_d &:= \{d,s_3,s_4\},\\
      S_y &:= \{y,s_5,s_6\}& S_z &:= \{z,s_7,s_8\}\\
      B_T &\mathrlap{\text{ is the set of all possible binary hyperedges on $T$},}\\
      T_a &:= \{a,t_1,t_2\}& T_b &:= \{b,t_3,t_4\},\\
      T_y &:= \{y,t_5,t_6\}& T_z &:= \{z,t_7,t_8\}\\
      E_1 &:= \{s_8,x_1\},\\
      E_{3i} &:= \{a,c,y,x_{3i}\}&&\mathrlap{\text{for } 1\le i<2n,}\\
      E_{3i+1} &:= \{b,d,z,x_{3i+1}\}&&\mathrlap{\text{for } 1\le i<2n,}\\
      E_{6n} &:= \{x_{6n},t_1\},\\
      E_{6i+1,6i+2} &:= \{a,b,x_{6i+1},x_{6i+2}\} &&\mathrlap{\text{for } 0\le i<n,}\\
      E_{6i+4,6i+5} &:= \{c,d,x_{6i+4},x_{6i+5}\} &&\mathrlap{\text{for } 0\le i<n, \text{ and}}\\
      E_{3i,3i+1} &:= \{x_{3i},x_{3i+1}\} &&\mathrlap{\text{for } 1\le i< 2n.}
  \end{align*}}%
  The set of boundary hyperedges is $\{A,B,C,D,\allowbreak S_c,S_d,
  \allowbreak S_y,S_z,\allowbreak T_a,T_b,\allowbreak T_y,T_z\}$.  The
  set of boundary vertices is $S\cup T$.  They are labeled from $1$ to
  $28$ in this order and by increasing indices. See \autoref{fig:ghtw}
  for an illustration of~$H_2$ induced on~$V \setminus (S\cup T)$.
\end{construction}

\tikzset{var/.style={inner sep=.15em,circle,fill=black,draw},
         hyperedge/.style={minimum size=1mm,rectangle,fill=white,draw},
         label distance=-2pt}

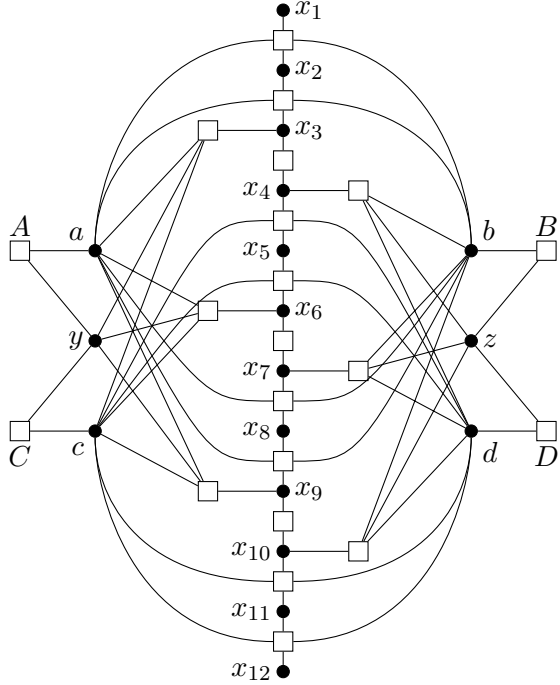
\begin{figure}[tb]
  \centering
  \begin{tikzpicture}[xscale=1,yscale=0.8]
    \node (x1) at (0,7) [var,label=right: $x_1$] {};
    \node (x2) at (0,6) [var,label=right: $x_2$] {};
    \node (x3) at (0,5) [var,label=right: $x_3$] {};
    \node (x4) at (0,4) [var,label=left: $x_4$] {};
    \node (x5) at (0,3) [var,label=left: $x_5$] {};
    \node (x6) at (0,2) [var,label=right: $x_6$] {};
    \node (x7) at (0,1) [var,label=left: $x_7$] {};
    \node (x8) at (0,0) [var,label=left: $x_8$] {};
    \node (x9) at (0,-1) [var,label=right: $x_9$] {};
    \node (x10) at (0,-2) [var,label=left: $x_{10}$] {};
    \node (x11) at (0,-3) [var,label=left: $x_{11}$] {};
    \node (x12) at (0,-4) [var,label=left: $x_{12}$] {};

    \node (e12) at (0,6.5) [hyperedge] {};
    \node (e23) at (0,5.5) [hyperedge] {};
    \node (e34) at (0,4.5) [hyperedge] {};
    \node (e45) at (0,3.5) [hyperedge] {};
    \node (e56) at (0,2.5) [hyperedge] {};
    \node (e67) at (0,1.5) [hyperedge] {};
    \node (e78) at (0,0.5) [hyperedge] {};
    \node (e89) at (0,-0.5) [hyperedge] {};
    \node (e910) at (0,-1.5) [hyperedge] {};
    \node (e1011) at (0,-2.5) [hyperedge] {};
    \node (e1112) at (0,-3.5) [hyperedge] {};
    
    \node (e3) at  (-1,5) [hyperedge] {};
    \node (e4) at  ( 1,4) [hyperedge] {};
    \node (e6) at  (-1,2) [hyperedge] {};
    \node (e7) at  ( 1,1) [hyperedge] {};
    \node (e9) at  (-1,-1) [hyperedge] {};
    \node (e10) at  ( 1,-2) [hyperedge] {};
    \node (a) at (-2.5,3) [var,label=above left: $a$] {};
    \node (b) at ( 2.5,3) [var,label=above right: $b$] {};
    \node (c) at (-2.5,0) [var,label=below left: $c$] {};
    \node (d) at ( 2.5,0) [var,label=below right: $d$] {};
    \node (y) at (-2.5,1.5) [var,label=left: $y$] {};
    \node (z) at ( 2.5,1.5) [var,label=right: $z$] {};
    \node (A) at (-3.5,3) [hyperedge,label=above:$A$] {};
    \node (B) at ( 3.5,3) [hyperedge,label=above:$B$] {};
    \node (C) at (-3.5,0) [hyperedge,label=below:$C$] {};
    \node (D) at ( 3.5,0) [hyperedge,label=below:$D$] {};
    
    \draw (a)--(A)--(y) (b)--(B)--(z) (c)--(C)--(y) (d)--(D)--(z);
    \draw (x1)--(e12)--(x2)--(e23)--(x3)--(e34)--(x4)--(e45)--(x5)--(e56)--(x6)--(e67)--(x7)--(e78)--(x8)--(e89)--(x9)--(e910)--(x10)--(e1011)--(x11)--(e1112)--(x12);
    \draw (x3)--(e3)--(a) (y)--(e3)--(c) (x4)--(e4)--(b) (z)--(e4)--(d) (x6)--(e6)--(a) (y)--(e6)--(c) (x7)--(e7)--(b) (z)--(e7)--(d) (x9)--(e9)--(a) (y)--(e9)--(c) (x10)--(e10)--(b) (z)--(e10)--(d);
    \draw (e12) .. controls +(-2,0) and +(0,1) .. (a);
    \draw (e23) .. controls +(-2,0) and +(0,1) .. (a);
    \draw (e12) .. controls +(2,0) and +(0,1) .. (b);
    \draw (e23) .. controls +(2,0) and +(0,1) .. (b);
    \draw (e45) .. controls +(-1,0)  .. (c);
    \draw (e45) .. controls +(1,0)  .. (d);
    \draw (e56) .. controls +(-1,0)  .. (c);
    \draw (e56) .. controls +(1,0)  .. (d);
    \draw (e78) .. controls +(-1,0)  .. (a);
    \draw (e78) .. controls +(1,0)  .. (b);
    \draw (e89) .. controls +(-1,0)  .. (a);
    \draw (e89) .. controls +(1,0)  .. (b);
    \draw (e1011) .. controls +(-2,0) and +(0,-1) .. (c);
    \draw (e1011) .. controls +(2,0) and +(0,-1) .. (d);
    \draw (e1112) .. controls +(-2,0) and +(0,-1) .. (c);
    \draw (e1112) .. controls +(2,0) and +(0,-1) .. (d);

  \end{tikzpicture}
  \caption{The incidence graph of~$H_2$ induced on~$V \setminus (S\cup T)$. Boxes represent hyperedges.}
  \label{fig:ghtw}
\end{figure}

\noindent We first give an outline of the remaining proof. Consider a
tree decomposition for~$H_n \oplush H_m$ with generalized hypertree
width~$4$. The aim is to prove~$n=m$. The vertex sets~$S$ and~$T$ and
the hyperedges containing them make sure that some bag~$\cB_S$ of the
decomposition contains all of~$S$ and that some bag~$\cB_T$ contains
all of~$T$.  Now, both in $\mathcal{G}(H_n)$ and in
$\mathcal{G}(H_m)$, there is a path from a vertex in $S$ to a vertex
in $T$ passing through all vertices $x_i$ by increasing indices. The
edges of this path are covered by intermediate bags lying on the path
from~$\cB_S$ to~$\cB_T$ on the tree decomposition.  Observe that no
vertex~$x_i$ is contained in a boundary hyperedge.  Therefore, when we
restrict the tree decomposition to the vertices in $H_n$, we recover a
tree decomposition for $H_n$ where all intermediate bags are covered
by at most $3$ hyperedges.  Moreover, our construction makes sure that
when a bag is covered by $3$ hyperedges, at least $2$ of them are
boundary hyperedges.  In every such tree decomposition for $H_n$, when
considering the intermediate bags starting from $\cB_S$ that contain
either $A,B$ or $C,D$ in their cover, we first encounter bags covered
by~$C,D$, then bags covered by~$A,B$, then bags covered by~$C,D$, and
so on, and there are exactly $n$~alternations from~$C,D$ to~$A,B$ in
this sequence.  Therefore, in order to be able to merge such
decompositions for~$H_n$ and~$H_m$, we must have~$n=m$.

We now give a more detailed proof of \autoref{THM:HTW}.  In the
construction of~$H_n$, the vertices in $S$ and $T$ and the hyperedges
containing them are only used to make sure that every tree
decomposition of $H_n$ with hypertree width~$4$ contains a bag
$\cB_{-1}$ with the vertices $S\cup \{c,d,y\}$ and a bag $\cB_{6n+1}$
with the vertices $T \cup \{b,y,z\}$.  Since the sets~$S\cup
\{c,d,y,z\}$ and $T \cup \{a,b,y,z\}$ can also be covered by
$4$~hyperedges, all of which are boundary hyperedges, let
$\mathcal{D}=(\{V_i : i\in I\},T)$ be a tree decomposition for~$H_n$
with the bags~$\cB_{-1} = S\cup \{c,d,y,z\}$ and $\cB_{6n+1} = T \cup
\{a,b,y,z\}$.  We observe that all other vertices of $H_n$ occur in
bags that are in the same connected component of the forest obtained
from $\mathcal{D}$ by removing these two bags.

\begin{cla}\label{cl:ghtw-bags}
  The tree decomposition $\mathcal{D}$ contains a bag $\cB_i$, $0\le
  i\le 6n$, with $\{s_8,x_1,\allowbreak c,d,\allowbreak y,z\}\subseteq
  B_0$, $\{t_1,x_{6n},a,b,y,z\} \subseteq B_{6n}$, and
  $\{a,b,c,d,y,z,\allowbreak x_i,x_{i+1}\} \subseteq \cB_i$, for every $i$, $1\le
  i < 6n$.
\end{cla}
\begin{proof}
 The primal graph $\mathcal{G}(H_n)$ contains the cliques
 $\{s_8,x_1\}$, $\{a,b,x_1,x_2\}$, $\{a,b,$ $x_2,x_3\}$, $\{a,c,y,x_3\}$, $\{x_3,x_4\}$, $\{x_4,b,d,z\}$, $\{c,d,x_4,x_5\}$, $\{c,d,x_5,x_6\}$, $\{a,c,y,x_6\}$, $\{x_6,\allowbreak x_7\}$, $\{b,d,z,x_7\}$, $\{a,b,x_7,x_8\},$ $\dots,$ $\{x_{6n},t_1\}$,
 and every two consecutive cliques in this list intersect in at least one vertex.
 In particular, we observe the path $(s_8,x_1,x_2, \dots, \allowbreak x_{6n}, t_1)$ in $\mathcal{G}(H_n)$.
 Thus, $\mathcal{D}$ contains bags~$\cB_0 \supseteq \{s_8,x_1\}$, $\cB_{6n} \supseteq \{t_1,x_{6n}\}$, and $\cB_i \supseteq \{x_i,x_{i+1}\}$, $1\le i<6n$.
 Moreover, each $\cB_i$, $0\le i<6n$, contains $c,d,y,z$ since $\cB_{-1}$ contains $c,d,y,z$, $\cB_{6n-1}$ contains $c,d$, $\cB_{6n+1}$ contains $y,z$, and
 without loss of generality, we can assume the $\cB_i$, $0\le i\le 6n$, were chosen such that they are on the path from~$\cB_{-1}$ to $\cB_{6n+1}$ in $T$.
 Similarly, each $\cB_i$, $1\le i\le 6n$, contains $a,b$.
\end{proof}

\noindent
A tree decomposition for $H_n$ is a \emph{good} tree decomposition if
it contains the bags~$\cB_{-1}=S\cup \{c,d,y,z\}$
and~$\cB_{6n+1}=T\cup \{a,b,y,z\}$ and every bag except $\cB_{-1}$ and
$\cB_{6n+1}$ can be covered with at most $3$ hyperedges, and in case
such a bag is covered with exactly $3$ hyperedges, two of these
hyperedges are in the boundary. A \emph{good} cover for a good tree
decomposition is a cover for each bag according to the specifications
of a good tree decomposition.

\begin{cla}\label{cl:ghtw-path}
  If $\mathcal{D}$ is a good tree decomposition for $H_n$, then, for
  every $i$, $-1\le i\le 6n$, there is a path from the bag $\cB_i$ to
  the bag $\cB_{i+1}$ that avoids all the bags $\cB_j$, $j\in
  \{-1,\dots,6n+1\} \setminus \{i,i+1\}$.
\end{cla}
\begin{proof}
  Suppose the path from $\cB_i$ to $\cB_{i+1}$ passes through $\cB_j$
  with $j\in \{-1,\dots,$ $6n+1\} \setminus \{i,i+1\}$.  Since every
  bag on the path from $\cB_i$ to $\cB_{i+1}$ contains $\cB_i \cap
  \cB_{i+1}$, we have that $x_{i+1} \in \cB_j$.  But then
  $\{s_8,x_1,c,d,y,z,x_{i+1}\}\subseteq B_j$ (if $j=0$) or
  $\{t_1,x_{6n},a,b,y,z,x_{i+1}\} \subseteq B_{j}$ (if $j=6n$) or
  $\{a,b,c,d,y,z,x_j,x_{j+1},x_{i+1}\} \subseteq \cB_j$ (otherwise),
  implying that $B_j$ cannot be covered by two hyperedges and it
  cannot be covered by three hyperedges of which two are in the
  boundary.  \end{proof}

\begin{cla}\label{cl:ghtw-cover}
  In every good cover, $\cB_0$ is covered by $\{E_1,C,D\}$, $\cB_{6n}$
  is covered by $\{E_{6n},\allowbreak A,\allowbreak B\}$, and for every $i$, $1\le i< 6n$,
\begin{align*}
 \cB_i \text{ is covered by } 
 \begin{cases}
 \{E_{i,i+1},C,D\} &\text { if } i \equiv 1 \pmod 6,\\
 \{E_{i,i+1},C,D\} &\text { if } i \equiv 2 \pmod 6,\\
 \{E_{i},E_{i+1}\} &\text { if } i \equiv 3 \pmod 6,\\
 \{E_{i,i+1},A,B\} &\text { if } i \equiv 4 \pmod 6,\\
 \{E_{i,i+1},A,B\} &\text { if } i \equiv 5 \pmod 6, \text{ and}\\
 \{E_{i},E_{i+1}\} &\text { if } i \equiv 0 \pmod 6.
 \end{cases}
\end{align*}
\end{cla}
\begin{proof}
  The claim easily follows from Claim \ref{cl:ghtw-bags}.  
\end{proof}

\noindent
Suppose $\mathcal{D}$ is a good tree decomposition for $H_n$.  The
\emph{backbone} of $\mathcal{D}$ is the path~$P$ in~$T$ starting at
the bag~$\cB_{-1}$ and ending at the bag $\cB_{6n+1}$.  By
\autoref{cl:ghtw-path}, $P$ visits $\cB_{0}, \cB_{1}, \dots, \cB_{6n}$
in this order.  Let $P_{i,j}$ denote the subpath of $P$ starting at
$\cB_i$ and ending at $\cB_j$. %

\begin{cla}\label{cl:ghtw-nocoverAB}
  For every $i\in \{0,6,12,\cdots,6n-6\}$, no bag on $P_{i,i+3}$ is
  covered by a set of hyperedges $\mathcal{Q}$ with $A,B\in
  \mathcal{Q}$ in a good cover.
\end{cla}
\begin{proof}
  Consider a bag $\cB$ on $P_{i,i+3}$ and let $\mathcal{Q} \supseteq
  \{A,B\}$ be a cover for $\cB$.  The bag $\cB$ contains the
  intersection of two bags that are consecutive in the list
  $\cB_i,\cB_{i+1},\cB_{i+1},\cB_{i+3}$.  Therefore, at least one of
  $x_{i+1},x_{i+2},x_{i+3}$ is in $\cB$.  We also have that $c,d\in
  \cB$ since $c,d\in \cB_i \cap \cB_{i+3}$.  However, no hyperedge
  contains $x_{i+1},c,d$ or $x_{i+2},c,d$ or $x_{i+3},c,d$. Thus,
  $|\mathcal{Q}|\ge 4$, and therefore $\mathcal{Q}$ is not part of a
  good cover.  
\end{proof}

\begin{cla}\label{cl:ghtw-nocoverCD}
  For every $i\in \{3,9,15,\cdots,6n-3\}$, no bag on $P_{i,i+3}$ is
  covered by a set of hyperedges $\mathcal{Q}$ with $C,D\in
  \mathcal{Q}$ in a good cover.
\end{cla}
\begin{proof}
  The proof is symmetric to the proof of \autoref{cl:ghtw-nocoverAB}.
\end{proof}

\noindent
Consider a good cover of $\mathcal{D}$.  A \emph{switch} is an
inclusion-wise minimal subpath $(Y_i, \dots, Y_j)$ of the backbone of
$\mathcal{D}$ where $Y_i$ is covered by~$\mathcal{Q}_i$ with $C,D\in
\mathcal{Q}_i$ and $Y_j$ is covered by~$\mathcal{Q}_j$ with $A,B \in
\mathcal{Q}_j$.  The \emph{signature} of a good cover of~$\mathcal{D}$
is its number of switches.

\begin{cla}\label{cla:ghtw-sign}
  Each good cover of each good tree decomposition of $H_n$ has
  signature~$n$.
\end{cla}
\begin{proof}
  The claim follows from Claims \ref{cl:ghtw-cover},
  \ref{cl:ghtw-nocoverAB}, and \ref{cl:ghtw-nocoverCD}.
\end{proof}

\noindent
Due to \autoref{cla:ghtw-sign}, we can speak of the signature of
$H_n$ and the signature of a good tree decomposition of $H_n$ as the
signature of some good cover of such a tree decomposition.

\bigskip

\noindent Let $H=H_n$ and $H'=H_m$.  Consider a tree decomposition
$\mathcal{D}=(\{V_i : i\in I\},T)$ of~$H \oplush H'$ with generalized
hypertree width~$4$.  Without loss of generality, suppose the bags
$\cB_{-1}=S\cup \{c,d,y,z\}$ and $\cB_{6n+1}=T\cup \{a,b,y,z\}$ are
leafs of this decomposition and their neighboring bags contain both
copies of $x_1$ and $x_{6n}$, respectively.  Let $\mathcal{D}_{|_H}$
denote the restriction of $\mathcal{D}$ to $H$, i.e., it has the same
tree, but each bag is restricted to the vertices of $H$.

\begin{cla}\label{cl:ghtw-good}
 $\mathcal{D}_{|_H}$ is a good tree decomposition for $H$.
\end{cla}
\begin{proof}
  Consider a bag $\cB$ of $\mathcal{D}$ besides $\cB_{-1}$ and
  $\cB_{6n+1}$.  The bag $\cB$ contains a copy of some $x_i$ from
  $H'$.  This vertex is covered by some hyperedge from $H'$ that does
  not belong to the boundary.  Therefore, $\cB_{|_H}$ is covered by at
  most $3$ hyperedges.  Suppose $\cB_{|_H}$ is covered by exactly $3$
  hyperedges.  Then, the cover of $\cB_{|_{H'}}$ contains at most one
  hyperedge that does not belong to the boundary.  But, since each
  such hyperedge covers at most $2$ vertices among $\{a,b,c,d\}$, the
  cover of $\cB$ contains at least $2$ boundary hyperedges.  This
  proves the claim.  
\end{proof}

\noindent Symmetrically, $\mathcal{D}_{|_{H'}}$ is a good tree
decomposition for $H'$.  Since $\mathcal{D}_{|_H}$ and
$\mathcal{D}_{|_{H'}}$ have the same signature, we conclude that $n=m$
due to \autoref{cla:ghtw-sign}. This proves that the \crc{}
$\sim_{\kghtw4}$ does not have finite index over~$\hlarge$. To prove
\autoref{THM:HTW}, it remains to prove that it also has infinite index
over~$\hsmall$.

\begin{proof}[Proof of \autoref{THM:HTW}]
  We aim to apply \autoref{hisoclosed2}. First, we show that the
  constructed graphs~$H_n$ have \iw{} at most~12.  The graph~$H_n
  \setminus (S\cup T \cup \{a,b,c,d,y,z\})$ is a disjoint union of
  trees and therefore, has tree decomposition of width $1$. From this
  tree decomposition, we obtain a tree decomposition of width~$7$
  for~$H_n \setminus (S\cup T)$ by adding $\{a,b,c,d,y,z\}$ to each
  bag of the decomposition.  Finally, we obtain a tree decomposition
  for $H_n$ of width at most 12 by adding the two bags~$S\cup
  \{a,b,c,d,y,z\}$ and $T\cup \{a,b,c,d,y,z\}$ and making them
  adjacent to arbitrary bags of the tree decomposition for $H_n
  \setminus (S\cup T)$.  To obtain a tree decomposition where one bag
  contains $\partial(H_n)$, we modify the tree decomposition of
  width~12 for the incidence graph~$H_n$ by adding the 28~boundary
  objects to each bag.  The result is a tree decomposition of width~40
  where $\partial(H_n)$~is contained in one bag. Then, we obtain a
  hypergraph~$H'_n$ from~$H_n$ by adding to~$H_n$ 13~additional,
  isolated boundary vertices. Clearly, $H_n$ and~$H'_n$ have the same
  generalized hypertree width. We use the tree decomposition of the
  incidence graph of~$H_n$ also for~$H'_n$, but we attach a bag
  consisting of~$\partial(H_n)$ and the 13~additional boundary
  vertices of~$H'_n$. This bag has width~$28+12=40$. Hence, the
  family~$H'_n$ we constructed from~$H_n$ has tree decompositions of
  width~40 of their incidence graphs such that one bag contains all
  41~boundary objects. Thus, \autoref{hisoclosed2} applies to our
  family.  
\end{proof}

\paragraph{Other width measures for hypergraphs.} \autoref{THM:HTW}
easily applies also to the problems \textsc{Hypertree Width} and
\textsc{Fractional Hypertree Width}, which asks whether a hypergraph
has (fractional) hypertree width at most~$k$. \textsc{Hypertree Width}
is W[2]-hard~\citep{GottlobGMSS05} with respect to~$k$ and
\textsc{Fractional Hypertree Width} is expected to be NP-hard for
constant~$k$~\citep{Marx10}.  Before discussing how \autoref{THM:HTW}
applies to these problems, we formally define these width measures.

\bigskip\noindent The \emph{hypertree width} of~$H$ is defined in a
similar way as the generalized hypertree width, except that,
additionally, the tree of the decomposition is rooted and a
hyperedge~$e$ can only be used in the cover of a bag~$V_i$ if
$V_i$~contains all vertices of~$e$ that occur in bags of the subtree
rooted at the node~$i$.

\looseness=-1 The \emph{fractional hypertree width} of~$H$ is also defined
similarly, except that it uses fractional covers: in a
\emph{fractional cover} of a bag, each hyperedge is assigned a
non-negative weight, and for each vertex in the bag, the sum of the
weights of the hyperedges incident to it is at least~$1$. The
\emph{fractional cover width} of the bag is the minimum total sum of
all hyperedges of a fractional cover.

\bigskip\noindent Let $\khtw k$ be the family of hypergraphs of
hypertree width at most~$k$ and $\kfhtw k$ be the family of
hypergraphs of fractional hypertree width at most~$k$.  To see that
the proof of \autoref{THM:HTW} applies to $\sim_{\khtw4}$, observe
that, in our construction, every hyperedge covering a bag is a subset
of that bag.  To see that it extends to $\sim_{\kfhtw4}$, observe that
for every bag $\cB_i$, $0\le i\le 6n$, an optimal fractional cover is
integral, and \autoref{cl:ghtw-nocoverAB} can be extended to $A,B\in
\mathcal{Q}$ with weight $1$---similarly for
\autoref{cl:ghtw-nocoverCD}.

\begin{corollary}
  \textsc{Fractional Hypertree Width} and \textsc{Hypertree Width} do
  not have finite index.
\end{corollary}

\paragraph{Indication for intractability.}\looseness=-1
Formally, \autoref{THM:HTW} only shows that a tree automaton cannot
decide the property of having constant \textsc{Hypertree Width} and,
by \autoref{nomso}, that this property is not expressible in monadic
second-order logic for hypergraphs.  However, following the
argumentation of \citet*{BFW92} for the \textsc{Triangulating Colored
  Graphs} problem introduced in \autoref{sec:intro} leads us to the
following conjecture.

\begin{conjecture}\label{conjec}
  \ghtw{} is W[1]-hard with respect to the parameter \iw.
\end{conjecture}

\noindent The reason for this conjecture lies in the number of
equivalence classes observed in the proof of \autoref{THM:HTW}, which
entails a lower bound on the amount of information that needs to be
maintained by an algorithm when it decides whether a given hypergraph
has (generalized, fractional) hypertree width~$k$ using a tree
decomposition of the incidence graph. Typical such algorithms
associate with each bag of the tree decomposition a table that is
computed from the tables associated with the tables of the child
bags. Observe that such an algorithm is essentially a tree automaton;
its states are the tables. Since the number of equivalence classes of
the \crc{} gives a lower bound on the number of states a tree
automaton needs to have in order to decide \ghtw{}, it also gives a
lower bound on the number different tables that have to be handled by
such a (simple) dynamic programming algorithm in order not to make
wrong decisions.

However, restricting the construction in the proof of
\autoref{THM:HTW} to graphs of at most $n$~vertices, the proof of
\autoref{THM:HTW} exhibits a class~$\mathcal{C}$ of $t$-boundaried
hypergraphs on at most~$n$ vertices with constant \iw{}~$t-1$, for
which the \crc{} has $\Omega(n)$ equivalence classes.  Now, consider a
class~$\mathcal{C}'$ of $O(k)$-boundaried hypergraphs where each
hypergraph contains $k$ copies of hypergraphs from~$\mathcal{C}$ and
has at most $n'$~vertices.  Then, the number of equivalence classes of
the \crc{} is~$\Omega((n'/k)^k)$ for $\mathcal{C}'$.  Hence, we
conjecture that an algorithm with running time~$f(k)\cdot n^c$ for a
constant~$c$ and a computable function~$f$ does not exist.

\section{Summary, Discussion and Open Problems}\label{sec:conc}
\noindent We extended the Myhill-Nerode theorem to hypergraphs, making
the methodology more widely applicable.  We did this in the general
framework of the hypergraph analogs of the \emph{large} and \emph{small}
universes of $t$-boundaried graphs.  We used the Myhill-Nerode approach
to obtain fixed-parameter linear-time algorithms for \textsc{Hypergraph
  Cutwidth} by
\begin{enumerate}[1)]
\item using the \emph{method of test sets} to prove ``finiteness''
  results in the \emph{large universe}, which is not only more
  convenient but also of independent interest in the context of
  communication complexity, and
\item then translating this result into an algorithm by means of a
  general machinery that links large universe results to the small
  universe of bounded treewidth representations.
\end{enumerate}
This approach has the advantage of being relatively simple and powerful
in the sense of~(1) yet relatively general in the sense of~(2).  One of
our principle objectives in this article was to establish powerful and
general methodologies to solve hypergraph problems.  If the principal
objective were to have the most efficient algorithm for
\textsc{Hypergraph Cutwidth}, then this is not the way to go.  That
would probably be to hunker down as tightly as possible into the small
universe, and do some serious dynamic
programming~\citep{BodlaenderK96,LagergrenA91,ThilikosSB05}.  Our
machinery only shows when such seriousness is worthwhile: for example,
it is for \textsc{Hypergraph Cutwidth}, but probably not for
\textsc{Hypertree Width} and its variants. 

There are many interesting open questions in this relatively
under-explored area.  The general theme is (referring to the very
abstract program about Myhill-Nerode equivalence classes): how do these
finiteness results relate to computational complexity?  How do these
finiteness issues translate between different settings?
\citet{CourcelleL96} proved some very interesting results about
translating finiteness results between the small and large universes
when the property ${\cal P}$ is restricted to simple graphs of bounded
treewidth.  They showed that
\begin{enumerate}[1)]
\item if $\mathcal{P}$ is restricted to simple graphs of treewidth at
  most~$t$, then finite index of~$\sim_{\mathcal P}$
  over~$\usmallp$ implies finite index over to
  $\ulargep$, and also that
\item (Rephrasing their Fact 9.4:) there is a property~$\mathcal{P}$ (of
  unbounded treewidth), such that for all~$t$, the canonical
  Myhill-Nerode equivalence relation~$\sim_{\mathcal P}$ has finite index over
  $\usmallp$, but infinite index on $\ulargep$.
\end{enumerate}
Do those results extend to hypergraphs in the representation framework
we have explored here?

Another question one might ask is whether there might be general
methods and meta-theorems for obtaining XP-complexity results, based
on communication complexity results with respect to~$\ulargep$.

\pagebreak
{\footnotesize\bigskip\noindent\textbf{Acknowledgments.} The authors are thankful to Mahdi Parsa
for fruitful discussions.  René van Bevern acknowledges support by the
Deutsche Forschungsgesellschaft (DFG), project DAPA (NI~369/12).  Rod
Downey acknowleges support by a grant from the New Zealand Marsden
Fund.  The remaining three authors acknowledge support by the
Australian Research Council, grants DP~1097129 (Michael R.\ Fellows),
DE~120101761 (Serge Gaspers), and DP~110101792 (Michael R.\ Fellows
and Frances A.\ Rosamond). NICTA is funded by the Australian
Government through the Department of Communications and the Australian
Research Council through the ICT Centre of Excellence Program.\par}

\appendix
\section{Appendix: Correction to the Myhill-Nerode theorem for graphs}\label{apxfix}}

Both parameterized complexity books of \citet{DF99,DF13} contain
incorrect proofs of the Myhill-Nerode theorem for graphs, which is
stated there as follows (the necessary definitions are briefly given in
the following; for more details, we refer to \citet[Section~12]{DF13}):

\begin{dftheorem}[{\citet[Theorem~12.7.2]{DF13}}]
  Let $F$~be a family of graphs.  Then, the following are equivalent:
  \begin{enumerate}[(i)]
  \item $F$ is $t$-finite state.
  \item $F_t$ is the union of equivalence classes of a right congruence
    of finite index over~$\mathcal U^{\text{small}}_t$.
  \item $\sim_{F_t}$ has finite index over $\mathcal
    U^{\text{small}}_t$.
  \end{enumerate}
\end{dftheorem}

\noindent Herein, $F_t=F\cap\mathcal U^{\text{small}}_t$. Moreover,
``$F$ is $t$-finite state'' means that the parse trees corresponding to
graphs in~$F_t$ are recognizable by a finite tree automaton.  The flaw
is in the proof of (iii)${}\rightarrow{}$(i).  We describe the flaw and
the necessary corrections using the terminology used in the book:
\begin{enumerate}[i)]
\item For a parse tree~$T$ of a \bndg{t}, $G(T)$~is the graph generated
  by~$T$.

\item The set~$L$ is the collection of parse trees generating the graphs
  in~$F_t$.
\item For a rooted labeled tree~$T$ with exactly one leaf labeled~$x$,
  $T'\cdot_x T$~is the rooted labeled tree obtained from identifying the
  root of~$T'$ with the leaf of~$T$ labeled~$x$, where the label~$x$ is
  replaced by the label of the root of~$T'$.
\end{enumerate}
Now, the approach of showing (iii)${}\rightarrow{}$(i) of \citet{DF13}
is as follows: if $F$~is not $t$-finite state, then there is an infinite
family of parse trees~$T_1,T_2,\dots$ such that, for each
pair~$T_i,T_j$, there is a parse tree~$T_{ij}$ with exactly one leaf
labeled~$x$ such that $T_i\cdot_x T_{ij}\in L\iff T_j\cdot_xT_{ij}\notin
L$.  The goal of \citet{DF13} is finding a parse tree~$Q_{ij}\in\mathcal
U^{\text{small}}$ such that~$G(T_k\cdot_xT_{ij})\cong G(T_k\oplusc
Q_{ij})$ for $k\in\{i,j\}$.  Since this means $G(T_i)\oplusc
G(Q_{ij})\in F_t\iff G(T_j)\oplusc G(Q_{ij})\notin F_t$, this shows that
$G(T_i)$ and~$G(T_j)$ are nonequivalent under~$\sim_{F_t}$ and, thus,
that $\sim_{F_t}$~has infinite index over $\mathcal U^{\text{small}}_t$.

\paragraph{Flaw.}\looseness=-1 \citet{DF13} obtain~$Q_{ij}$ using a so-called Parsing
Replacement Property, which allows replacing parsing operators of the
parse tree by the operator~$\oplusc$ while maintaining isomorphism with
the graph to be parsed.  However, since applying the Parsing Replacement
Property may alter the boundary, its application might violate
isomorphism if applied to the \emph{inner} nodes of the parse tree, like
it is done in the proof.%

\paragraph{Correction.} Despite the flaw, $Q_{ij}$ is easy to
construct for the parsing operators used by \citet{DF13} (precisely
those in \autoref{ourops} for~$\cmax=1$) as follows.  The parse
tree~$T_{ij}$ corresponds to a tree decomposition for the
graph~$G(\emptyset_1\cdot_x T_{ij})$: simply associate each node of
the parse tree with a bag that contains exactly the vertices that are
labeled at that parse tree node.  By merely choosing as the root of
this tree decomposition the bag corresponding to the node labeled~$x$
in~$T_{ij}$, applying the procedure by \citet[Theorem~12.7.1]{DF13},
and permutating the boundary labels accordingly, one obtains from this
tree decomposition a parse tree~$Q_{ij}\in\mathcal U^{\text{small}}$
such that $G(T_k\cdot_xT_{ij})\cong G(T_k\oplusc Q_{ij})$ for
$k\in\{i,j\}$ (when ignoring the boundary).

{\small\setlength{\bibsep}{0pt}
\bibliographystyle{abbrvnat}
\bibliography{trellis}}

\end{document}